\documentclass[12pt]{article}

\usepackage[utf8]{inputenc}
\usepackage{tikz}
\usepackage{color}
\usepackage{amsmath}
\usepackage{amsfonts}
\usepackage[title]{appendix}
\usepackage{stmaryrd}
\usepackage{amssymb}
\usepackage{nicefrac}
\usepackage{comment}

\makeatletter
\@ifpackageloaded{natbib}{}{%
  \usepackage{natbib}%
}
\@ifpackageloaded{hyperref}{}{%
  \usepackage{hyperref}%
}
\makeatother

\usepackage{enumerate}
\usepackage[normalem]{ulem}
\usepackage[lmargin= 1.2in,rmargin= 1.2in,tmargin= 1.2in,bmargin= 1.2in]{geometry}
\usepackage{xcolor}
\usepackage{thmtools}
\usepackage{thm-restate}
\usepackage[textwidth=15cm]{todonotes}

\newtheorem{theorem}{Theorem}

\newtheorem{claim}{Claim}

\newtheorem{corollary}{Corollary}

\newtheorem{definition}{Definition}

\newtheorem{lemma}{Lemma}

\newtheorem{proposition}{Proposition}
\newtheorem{remark}{Remark}

\newenvironment{proof}[1][Proof]{\textbf{#1.} }{\ \rule{0.5em}{0.5em}}

\newcommand{\ep}{\varepsilon}
\newcommand{\dN}{{{\bf N}}}

\newcommand{\E}{{{\bf E}}}

\newcommand{\prob}{{{\bf P}}}

\DeclareMathSymbol{\mlq}{\mathord}{operators}{``}
\DeclareMathSymbol{\mrq}{\mathord}{operators}{`'}

\def\r{\rho}

\def \b{\textcolor{blue} }
\def \r{\textcolor{red} }

\newcommand{\rs}[1]{\textcolor{red}{\sout{#1}}}

\begin{document}

\title{Strategic Experimentation with Private Payoffs%
\thanks{Renault  acknowledges funding from ANITI,    grant ANR-19-PI3A-0004, and from the ANR under the Investments for the Future program, grant ANR-17-EURE- 0010. 
Solan acknowledges the support of the Israel
Science Foundation, Grant \#211/22. 
Vieille thanks the HEC Foundation for support.}}
\author{J\'er\^ome Renault\thanks{Mathematics and Statistics Department, Toulouse School of Economics, Toulouse, France. E-mail: \textsf{jerome.renault@tse-fr.eu}.},
Eilon Solan\thanks{School of Mathematical Sciences, Tel Aviv University, Tel
Aviv 69978, Israel. E-mail: \textsf{eilons@tauex.tau.ac.il}.},
and Nicolas Vieille\thanks{Department of Economics and Decision Sciences, HEC Paris, 1, rue de
la Lib\'{e}ration, 78 351 Jouy-en-Josas, France. E-mail: \textsf{vieille@hec.fr}.}}
\date{\today}

%
\maketitle
\begin{abstract}
We study a strategic experimentation game with exponential bandits, in which experiment outcomes are private. The equilibrium amount of experimentation is always higher than in the benchmark case where experiment outcomes are publicly observed. In addition, for pure equilibria, the equilibrium amount of experimentation is at least socially optimal, and possibly higher. We provide a tight bound on the degree of over-experimentation. The analysis rests on a new form of encouragement effect, according to which a player may hide the absence of a success to encourage future experimentation by the other player, which incentivizes current experimentation. 

\end{abstract}



\section{Introduction}

In many dynamic economic environments, decision makers face uncertainty about the profitability of different strategies and must balance between exploring new opportunities and exploiting known ones. 

This tension is central to strategic experimentation, where multiple players engage in learning through costly actions while observing each other’s choices.
Classical examples include firms investing in uncertain technologies, policymakers testing new policies, or investors navigating uncertain markets.  Decision makers 
gain information through their own experience
as well as by observing the information gained by others, or 
from their behavior if this information cannot be directly observed.

Starting with \cite{rothschild1974two} for one-player problems, and with \cite{bolton1999strategic} for strategic problems, 
 situations involving strategic experimentation have been analyzed in two-arm bandit settings, see \cite{bergemann2008} for a survey.
In such settings, each  player is facing a repeated choice between a safe arm and a risky arm, which is either good or bad.

In  some applications, experimentation outcomes are best thought of as private.
Yet the literature has largely assumed that both actions and experimentation outcomes are public, in which case all players share the same belief on the risky arm's type. Exceptions include \cite{rosenberg2007social}, \cite{murto2011learning}, \cite{heidhues2015strategic}, and  \cite{bonatti2011collaborating}. Quoting 
 \cite{horner2017learning}, the 
 strategic experimentation problem with observed actions and unobserved outcomes remains largely unsolved. 

The present 
paper  contributes to this research direction. 
We adopt a discrete-time version of the workhorse model of exponential bandits, introduced in \cite{keller2005strategic} and used in most of the literature.  
Two players alternate over time in choosing one of the two arms. 
The safe arm is costless and delivers no payoff, while 
pulling the risky arm is 
costly and  yields a positive payoff at random times, only if its type is good. 
Actions are public, but outcomes/payoffs are not. 

In the benchmark case where outcomes are also public, a robust finding from the existing literature is that the equilibrium amount of experimentation is socially suboptimal when time is discrete, see \cite{heidhues2015strategic}.%
\footnote{This is also true of Markov equilibria in continuous time. Going beyond Markov equilibria,  
 \cite{keller2005strategic} for exponential bandits, and \cite{horner2022overcoming} more generally, devise strategies that overcome free-riding.}  Agents do not internalize the social value of experimentation, and underinvest in experimentation.

When outcomes are private, a player is unsure whether to interpret the experiments of other players as evidence that their previous experiments were successful, or as pure experimentation. This opens the way for a potential manipulation of beliefs, with  player $i$  experimenting to convince player $j$ that he was successful, 
thereby induce $j$ to experiment, 
with the intent of free-riding on $j$'s 
experiments. In other words, \emph{today}'s experimentation may serve as a signal that one may have been successful. This induces future experimentation by the other player, 
which encourages today's experimentation. This effect is similar to, but different from, the \textit{encouragement effect} introduced by \cite{bolton1999strategic}, which they defined as  \textit{the prospect of future experimentation by others encourages agents to increase current experimentation.} 
As we show, this effect implies that  the equilibrium amount of experimentation is always higher when outcomes are private rather than public. 
In  equilibrium, though, $j$'s inferences account for the possibility that $i$'s experiments may be deceptive:
equilibrium inferences are highly complex and involve beliefs of all orders.

 In the one-player case and in the two-player public case, the exploration/exploitation trade-off between the option value of an additional experiment and the opportunity cost of that experiment is captured via a single belief cut-off $p^*$, which dictates the optimal behavior. When outcomes are private, we introduce a new \emph{encouragement cut-off} $\widehat p$ that
plays a comparable role, and that balances the cost of one extra experiment with the value of \emph{two} experiments.\footnote{With the outcome of the second experiment being disclosed only later.}
 As we show, players do not experiment when their belief is below $\widehat p$, and often 
experiment when it is above
$\widehat p$. 
These findings allow us to bound the equilibrium amount of experimentation
for pure equilibria: in the absence of conclusive news, the amount of experimentation is  at least the socially optimal amount, and at most twice this amount. Equilibria do not display under-experimentation, but 
 may display over-experimentation.

\medskip
The related literature is discussed in Section \ref{sec lit}. 
The paper most closely related  to our's is Heidhues et al.~\cite{heidhues2015strategic}, which looks at a similar model but assumes that  direct, costless communication between the players is possible at any point in time. Since players have no incentive to lie once they know that the risky arm is good, truthful communication is feasible, and a chief question is to what extent truthful disclosure should be postponed.%
\footnote{A further notable difference lies in the solution concept. Once successful, it is strictly dominant for a player to pull repeatedly the risky arm. Yet, the sequential equilibria in Heidhues et al.~\cite{heidhues2015strategic}   share the feature that  ``a player who was made subjectively certain of the good state by an opponent's announcement of a success maintains this belief even in the face of the opponent's subsequent use of the safe arm.''
Such a behavior in equilibrium may be deemed undesirable. We use a refinement of sequential equilibrium that excludes it.}

\medskip

The  paper is organized as follows. The model is set up in Section \ref{sec model}. Section \ref{sec public payoffs} discusses the benchmark case in which  experiment outcomes are public. Our  equilibrium concept is defined in Section \ref{sec concept}.
 Section 
\ref{sec: encouragement} introduces the encouragement cut-off, with a few examples. Our main results are stated in Sections \ref{sec first bench} and \ref{sec main results}, and proven in the Appendix, Sections \ref{section:proof:prop private} through \ref{sec proof th private3}.\footnote{Additional results are provided in the supplementary material.}

\section{Model}\label{sec model}

We consider a discrete-time version of a strategic experimentation game, see \cite{keller2005strategic}. There are two available arms, a safe arm  $S$ and a risky one $R$. The safe arm always delivers the same payoff, which we normalize to zero. The risky arm is either good ($G$) or bad ($B$), and entails an opportunity cost 
$c>0$.
The type $\theta\in\{G,B\}$ of the risky arm is drawn at time $t=0$ and remains fixed throughout the game. We denote by $p_0:=\prob(\theta=G)>0$ the common prior distribution of $\theta$. 

The risky arm always delivers a payoff of zero in the bad state $\theta= B$. 
If $\theta= G$, whenever selected, the risky arm  yields a random payoff of either 0 or $m$, with probabilities $1-\lambda$ and $\lambda\in (0,1)$ respectively.

Over time, two players alternate in choosing one of the two arms, with player 1 acting first. For convenience, we define a \emph{period} as consisting of two consecutive choices, one for each player. 
That is, in each period $t\geq 1$,
first player 1  chooses an arm, and then player 2 chooses an arm. 
The players choices are \emph{publicly} observed but payoffs are \emph{private}.%
\footnote{Payoffs are private throughout the paper, except in Section \ref{sec public payoffs}.} 
The common discount factor between consecutive periods is $\delta\in (0,1)$. 
This completes the description of the model. 

\medskip

If $\lambda m\leq c$, it is optimal to always pull the safe arm, even if the risky arm is known to be good ($\theta = G$). We rule out this trivial case and assume throughout that $g:=\lambda m-c> 0$. 

We say that a player \emph{experiments} when he chooses $R$, and that the experiment is \emph{successful} if it yields a payoff of $m$.
In line with most of the strategic experimentation literature, news are conclusive:  
when  a player 
with current belief $p$ experiments,\footnote{Unless specified, we mean the first-order belief of that player,  identified with the belief assigned to $\theta = G$.} his belief jumps to 1 if  successful, and decreases to $\displaystyle\phi(p):=\displaystyle \frac{p(1-\lambda)}{p(1-\lambda) +1-p}$ if not.

Since $g> 0$,
a  player  finds it optimal to choose repeatedly $R$ once successful. 
However, since only choices are observed, 
this success is not observed by the other player, who has to draw inferences from the choices  of the successful player. 
The endogenous probabilistic inferences from such observational learning are the topic of the paper. 
\medskip

The assumption that payoffs are private is the main element that  differentiates our paper from  the existing literature. The assumption that players alternate in making choices is mostly for convenience.
However,
the assumption that time is discrete plays a critical role in  shunning cheap talk opportunities. 
Indeed, in continuous time, players have the option to switch back and forth in the fraction of an instant, allowing them  to encode their information at no cost, thereby providing cheap talk opportunities.

\section{The case of public payoffs}\label{sec public payoffs}

In this section (only), we assume that both payoffs and choices are public. 
This well-studied case is a natural benchmark. The results in this section are to a large extent well known, so we remain slightly informal.

\subsection{The one-player case and the social optimum}
\label{section one-player}

Assume there is a single player with discount factor $\delta$. 
The optimal policy is unique up to ties and is a cut-off policy: in a given period, the optimal choice is to experiment 
if and only if the current belief exceeds 
some  cut-off value  $p^*_\delta$. The value of $p^*_\delta$ is pinned down by the condition that when holding the belief $p^*_\delta$, the agent is indifferent between choosing the safe arm forever, and experimenting one last time. This indifference condition translates to  $(1-\delta)\left(p^*_\delta \lambda m -c\right) +\delta  p^*_\delta  \lambda g=0,$
that is, 
\begin{equation}
\label{def:p*}
 p^*_\delta= \frac{c(1-\delta)}{c(1-\delta)+g(1-\delta(1-\lambda))}.
\end{equation}

It is thus optimal to pull the risky arm $N^*_\delta:=\inf\{n \geq 0 : \phi^n(p_0)<p^*_\delta\}$ times before  switching forever to the safe arm if all experiments failed. 
We  write $N^*$ and $p^*$ for $N^*_\delta$ and $p^*_\delta$ when the discount factor is clear from the context.

\medskip

Assume now that there are two agents, the strategies of which are dictated by a social planner whose objective is 
to maximize joint payoffs. If the  goal of the social planner is to maximize the weighted sum  $\gamma^1+\sqrt{\delta} \gamma^2$, where $\gamma^i$ is the $\delta$-discounted sum of player $i$'s payoffs  over time, then the problem of the social planner reduces to that of a single agent choosing an action in each \emph{half}-period, with a  discount factor of $\sqrt{\delta}$.
Accordingly, in this social optimum, agents alternate in pulling the risky arm until the common belief falls below $p^{**} =  p^{**}_\delta:= p^*_{\sqrt{\delta}}$, and the optimal number of experiments is $N^{**} =  N^{**}_\delta:=N^*_{\sqrt{\delta}}$.

\subsection{The two-player case}

When payoffs are public, the experimentation game may be reinterpreted as a stochastic game with perfect information, where the state variable is the common belief over  $\theta$.
 Given this interpretation, a \emph{Markov} strategy is a function $f : [0,1] \to \Delta(\{S,R\})$,
with the understanding that $f(p)$ is the mixed move  selected when the current belief is $p$. 
A \emph{symmetric Markov equilibrium} is a Markov strategy $f$, such that the profile $(f,f)$ is a subgame perfect equilibrium (henceforth, SPE).

Proposition \ref{prop public} below summarizes the main results for that case. In this statement, 
the role of the assumption on $p_0$ is to rule out non-generic cases, and  $N_{\mathrm e} \leq +\infty$ is the total (random) number of experiments over time. We recall that the $n$-th iterate $\phi^n(p_0)$ is the common belief after $n$ failed experiments.

\begin{proposition}[Public payoffs] \label{prop public}
Assume that  $\phi^n(p_0)\neq p^*$ for each $n$.
The following properties hold:
\begin{description}
    \item[P1] At any Nash equilibrium,  $N_{\mathrm e}=N^*$ with probability one conditional on $\theta=B$.
    \item[P2] There is a pure $\mathrm{SPE}$, in which players experiment $N^*$ times in a row, and then stop experimenting if unsuccessful.
    \item[P3] There exists a unique symmetric Markov equilibrium, but a continuum of $\mathrm{SPE}$ payoffs if $p_0>p^*$.
\end{description}

\end{proposition}

If $\theta=G$, agents behave as when the state is $\theta=B$ until one experiment is successful. That is, $N_{\mathrm e}=+\infty$ 
if one of the first $N^*$ experiments is successful, and $N_{\mathrm e}= N^*$ otherwise. 

According to \textbf{P1}, which mirrors Proposition 1 in \cite{heidhues2015strategic}, the number of experiments is the same across all equilibria, and is inefficiently low since $p^{**}<p^*$ implies $N^*\leq N^{**}$. Since the assertion holds for
all Nash equilibria, the scope of the result is broad. According to \textbf{P3} however, while the equilibrium number of experiments is uniquely pinned down, their timing is not.

Since Proposition \ref{prop public} is in line with existing results, 
its proof is in the supplementary material.
\medskip

\begin{remark}\label{remark1}
Define the \emph{myopic} cutoff as $p_{\mathrm{myop}}:=\frac{c}{\lambda m}= \frac{c}{c+g}$, and note that for $p>p_{\mathrm{myop}}$, the risky arm is more informative and produces a higher current reward than the safe arm.

Let $\sigma$ be the Markov profile in which the active player experiments if and only if the current belief is above $p^*$. We note that $\sigma$ is not an SPE if $p_0\in (p^*,\min(p_{\mathrm{myop}},\phi^{-1}(p^*))$. The reason is that given $\sigma$, player $\mathrm{1}$ faces the following trade-off at the initial period. If he chooses $R$ as prescribed by $\sigma$, no one will experiment further, since $\phi(p_0)<p^*$. If he deviates to $S$, player $\mathrm{2}$ will experiment once, because the current belief will still be $p_0$. That is, deviating to $S$  avoids paying the cost of experimenting and does not affect the amount of experimentation. Since $p_0< p_{\mathrm{myop}}$, deviating to $S$ is profitable.%
\footnote{If instead $p_0\in [p_{\mathrm{myop}},\phi^{-1}(p^*)]$ (which requires $\phi^{-1}(p^*)\geq p_{\mathrm{myop}}$), then $\sigma$ is an SPE.} 

\end{remark}

\section{The private case: a first benchmark}\label{sec first bench}

\subsection{Strategies}

We introduce some terminology. \emph{Full histories} specify entirely the play up to the current period, that is, the realized state $\theta$, the sequence of past  choices, and all experiment outcomes. \emph{Public histories} list past choices and are therefore finite sequences of $R$'s and $S$'s. We denote by $H:=\{R,S\}^{<\mathbb{N}}$ the set of public histories. Public histories play a leading role,  and we reserve the symbol $h$ to such histories. We note that player 1 (resp., player 2) is active at $h$ if and only if the length $|h|$ of $h$ is even (resp., odd).
\emph{Private histories} of player $i$ specify in addition whether $i$'s experiments were successful. A strategy  $\sigma^i$ of player $i$ maps the set of $i$'s private histories (where active) into the set $\Delta(\{R,S\})$ of mixed moves. 

Once a player is successful, his only sequentially rational continuation strategy is to choose repeatedly the risky arm. When constructing sequential equilibria, we will abstain from repeating that players choose $R$ at such private histories. To simplify notations, we thus define a strategy profile as a map $\sigma: H\to \Delta(\{S,R\})$, with the understanding that   $\sigma(h)$ is the mixed move at $h$ of the active player if   (i) the sequence of previous choices is $h$ and (ii) all the past experiments of this active player failed.\footnote{The choice of this concise notation will raise two issues. The first will arise when specifying off-path beliefs in a sequential equilibrium: when checking the consistency of beliefs and strategies, we will have to allow for completely mixed strategies that play $S$ with positive probability, even when successful. 
The second issue will arise when discussing Nash equilibria, since Nash equilibria may play dominated continuation strategies off-path. We will deal with these issues when needed, without introducing additional notations.} 
\medskip

There is one case where the optimal arm choice is clear. 
Assume the active player holds the belief  $p$, and chooses to experiment. His expected current reward is $p\lambda m -c$ and his expected continuation payoff is at most $pg$, irrespective of past play and of players' continuation strategies. 
If instead the active player chooses the safe arm, 
he can guarantee a non-negative payoff (e.g., by always choosing the safe arm).
Thus, it is dominating to choose the safe arm 
when $(1-\delta)\left(p\lambda m -c\right) +\delta pg< 0$, that is, when 
\begin{equation}\label{eq tildep}
p<\displaystyle \widetilde p:=\frac{c(1-\delta)}{(1-\delta)\lambda m +\delta g}.
\end{equation}  
Note that $\widetilde p < p^*$. 
We denote by $\widetilde N := \min\{n \geq 0 \colon \phi^{n}(p_0) < \widetilde p\}$ 
the number of failures needed to push the belief below $\widetilde p$.

\subsection{Private vs.~public payoffs: a first comparison}\label{sec first results}

Our first main result states that the equilibrium amount of experimentation is \emph{always} higher when payoffs are private.

\begin{theorem}[Private payoffs]\label{prop private}
Assume that $\phi^n(p_0)\neq p^*$ for each $n$. For any Nash equilibrium $\sigma$ of the game with private payoffs, one has $\prob_\sigma(N^* \leq N_{\mathrm e} < + \infty\mid \theta = B)=1$.
\end{theorem}

If $\theta= G$, the distribution over histories coincides with the distribution if $\theta =B$, as long as all experiments fail. If one experiment is successful, the successful player repeats $R$ forever, and the other player eventually sticks to $R$.

This  result is largely due to a  form of encouragement effect. If player $j$  sticks to the risky arm, player $i$ will put some weight on the possibility that $j$ was successful and  will therefore be more optimistic than if $j$'s failures were observed by $i$. As a result, $i$ will experiment longer ($N_{\mathrm e}\geq N^*$). 

This intuition leaves open the possibility that, by mutual encouragement, players might engage in an endless phase of experimentation, each player attributing the other player's insistence on $R$ as evidence that he 
(the other player) was successful, and their own failures to bad luck. As the second part of Theorem \ref{prop private} shows, this cannot be the case, and experimentation must eventually stop in the absence of conclusive evidence ($N_{\mathrm e}<+\infty$). 

\medskip

\begin{proof}[Proof Sketch] 
The formal proof appears in Section~\ref{section:proof:prop private}.
Fix a Nash equilibrium $\sigma$.\footnote{Our proof allows for Nash equilibria in which a player may pick $S$ after being successful.}
The proof that $N_{\mathrm e}\geq N^*$ with probability 1 builds on the above intuition. Fix any on-path public history $h$. If players knew that they  were both  unsuccessful  so far, they would share the belief 
 $\phi^{n_{\mathrm e}(h)}(p_0)$ (where $n_{\mathrm e}(h)$ is the number of experiments along $h$). Since outcomes are private, the belief $p^i(h)$ of each player $i$ 
 at $h$
 accounts both for  this worst-case scenario and for the possibility that $j$ may have been successful,  implying that $p^i(h)\geq \phi^{n_{\mathrm e}(h)}(p_0)$. 
The rest of the argument is quite similar to the proof for the public case. 
\medskip

To prove  that $N_{\mathrm e}<+\infty$ if $\theta= B$, we first note that at equilibrium players do not switch back and forth infinitely often between the two arms. The logic is the following. 
Assume there  is an on-path history $h$ that ends with $S$ and such that $n_{\mathrm e}(h) >\widetilde N$. 
The player active at $h$ infers that all experiments  so far
were failures;  therefore, he holds the belief $\phi^{n_{\mathrm e}(h)}(p_0)<\widetilde p$, and must play $S$. The same argument holds for the other player at the history $hS$.
Hence $\sigma(h)=\sigma(hS)= S$, and the equilibrium sequence of choices following $h$ must be $S^\infty$. 
This implies that, with probability 1, the sequence of choices  ends either with $S^\infty$, or with $R^\infty$. 

Hence, with probability 1, one of the two arms is used finitely many times. Assume that conditional on $ B$, there is a positive probability that $S$ is pulled only finitely many times. On that event,  each player $i$ eventually assigns a probability arbitrarily close to 1 to the event that the other player, $j$,  will always choose $R$ in the future, independently of the experiments' outcomes. Consequently, player $j$'s choices become increasingly non-informative from  $i$'s perspective. Eventually, the updating of $p^i(h)$ is mostly based on $i$'s failures, so that $p^i(h)\to 0$, implying that $i$ will eventually switch to $S$. 
\end{proof}

\section{The private case: conceptual issues}\label{sec concept}

This section presents and clarifies our solution concept. We start with a simple  Nash equilibrium which we  next modify into a pure, sequential equilibrium. 
As we will see, 
this sequential equilibrium  relies on off-path beliefs that assign probability 1 to a player making strictly dominated choices.
Following the earlier literature, we then explicitly rule out such beliefs and introduce the concept of \emph{reasonable} sequential equilibrium.

\subsection{A simple pure Nash equilibrium}
\label{section:simple:pure}

Constructing a pure Nash equilibrium is not difficult. 
Let $p_0\geq p^*$ be given, and recall that $N^*=\inf\{n \geq 0: \phi^n(p_0)<p^*\}$. 
Consider the infinite sequence of choices $h^*_\infty:= (RR)^{N^*}\!\!\!\cdot\! S^\infty$, along which each player experiments $N^*$ times, then switches to the safe arm.
Define $\sigma_0$ as the strategy profile that follows  the sequence $h^*_\infty$ as long as past choices are consistent with it
and that selects $R$ otherwise. That is, $\sigma_0(h)=S$ if and only if $h$ is a prefix of $h^*_\infty$ with length $|h|\geq 2 N^*$.

Put differently, according to $\sigma_0$, players  `report' truthfully their private information once they reach period $N^*+1$. 
If player 1 chooses $R$ in period $N^*+1$, then player 2 infers that $\theta= G$ and chooses $R$ ever after. 
If player 1 chooses $S$ in period $N^*+1$,
then in the second half of that period player 2 reports  his information by choosing either $S$ or $R$. An observable deviation from $\sigma_0$ triggers the other player to choose $R$ forever. 
The fact that $\sigma_0$ is an equilibrium  (see Section~\ref{appendix prop 5} for details) follows from two observations. During the experimentation phase, the fact that the other player keeps experimenting is uninformative, and each player's belief is therefore updated only on the basis of his own experiments.
Under $\sigma_0$, each player thus experiments as much as if he were alone, 
and only then learns the outcomes of $N^* $ experiments of the other player. 
On the other hand, if a player deviates and experiments for a suboptimal number  of periods, he foregoes in addition the opportunity to learn the outcomes of the other player's experiments. 

\medskip

According to the strategy profile $\sigma_0$, each player learns the other player's private information only once the experimentation phase is over. This positive externality can be used to incentivize players to experiment \emph{more}. 
Specifically, given $n\geq 0$, consider  the infinite sequence $h^*_\infty(n):= (RR)^{N^*+n}\!\cdot\! S^\infty$, and define $\sigma_n$ as the strategy profile that follows  $h^*_\infty(n)$ as long as the past sequence of choices is consistent with it, and selects $R$ otherwise.
 According to $\sigma_n$,  players are supposed to experiment for $n$ additional periods once their belief falls below $p^*$, rather than to stop immediately. The rationale is that these $n$ extra experiments will give them access \emph{for free} to the outcome of $N^* +n$ experiments (those performed by the other player), which they would never learn if they were to stop before. 

In Section \ref{appendix prop 5}, we identify a simple condition on $\lambda$, $\delta$, $p_0$, and $n$ under which the strategy profile $\sigma_n$ is a Nash equilibrium. This condition implies Proposition \ref{prop comparison p/e} below, according to which the amount of experimentation may be significantly higher than in the public case.

\begin{proposition}
\label{prop comparison p/e}
The following two statements hold.
\begin{itemize}
 \item There are  values of $(\lambda,\delta,p_0,n)$ such that $\sigma_n$ is a
 Nash
 equilibrium and the belief $p_f:= \phi^{n+N^*_\delta-1}(p_0)$ held \emph{before} 
 the last experimentation is such that $\displaystyle\frac{p_f}{p^*_{\delta}}$ is arbitrarily close to $\displaystyle\frac{1}{e}$.
 
    \item There are  values of $(\lambda,\delta,n)$ such that $\sigma_n$ is a
 Nash
 equilibrium when $p_0= p^*_\delta$ and $n$ is arbitrarily large. In that equilibrium, the ratio $\displaystyle \frac{N_{\mathrm e}}{N_\delta^*}= \frac{2(n+1)}{1}$ is arbitrarily large, while the ratio $\displaystyle \frac{N_{\mathrm e}}{N_\delta^{**}}$ is arbitrarily close to  $\displaystyle\frac{2 x_0}{\ln 2}\simeq 2.299$, where $x_0>0$ solves  $x+e^{-2x}=1$.
  
\end{itemize}

\end{proposition}

\subsection{Sequential equilibria}

A sequential equilibrium is a pair $(\sigma,\pi)$, where $\sigma$ is a  strategy profile and $\pi=(\pi^1,\pi^2)$ is a belief system, such that $\sigma$ is sequentially rational given $\pi$,
and $\pi$ is consistent with $\sigma$.\footnote{
Sequential equilibria were introduced for finite games. The definition readily extends to games with infinite horizon, see, e.g., \cite{fudenberg1983subgame}.
}
In our framework, a belief system of player $i$ is a collection $(\pi^i(\cdot\mid h^i))_{h^i}$, which associates with each \emph{private} history $h^i$ of player $i$ a probability distribution on the set of all full histories that are consistent with $h^i$. 
Given such a belief system, we let
\begin{itemize}
    \item $p^i(h)$ be the probability assigned by $i$ to $\theta= G$, conditional on (i) the sequence of past choices is $h$, and (ii) player~$i$ was not successful along $h$. It is the first-order belief of $i$, when not successful.
    \item $q^i(h)$ be the probability assigned by $i$ to the fact that player $j$ was successful, conditional on  (i) the sequence of past choices is $h$, and (ii) $\theta = G$. It is related to the  second-order belief of $i$.
\end{itemize}


\subsubsection{A pure sequential equilibrium}\label{sec pure SE}

Under the strategy profile $\sigma_0$ defined in Section~\ref{section:simple:pure}, 
any observable deviation triggers the indefinite choice of $R$ by the non-deviating player. We outline here the construction of a similar pure sequential equilibrium $(\sigma,\pi)$ in which an observable deviation of player $i$ is interpreted by $j$ as  `conclusive' evidence that $i$ was successful, even if that deviation consisted in choosing $S$, and triggers the indefinite choice of $R$ by $j$. Such off-path beliefs assign probability 1 to the event that the other player has made a strictly dominated choice. We later explicitly rule out such beliefs, hence we postpone  the detailed construction of $(\sigma,\pi)$  and the proof of optimality to the supplementary material.

Although the definition of $\sigma$ is inspired by $\sigma_0$,
the consistency requirement on beliefs prevents us from fully duplicating $\sigma_0$. Indeed, a deviation by $i$ cannot be interpreted as evidence that $i$ was successful, if $i$ has never experimented.

\paragraph{Outline.}
 On the equilibrium path,  $\sigma$ follows the  sequence $h^*_r:=(RS)^r\!\cdot (RR)^{N^*-r}\!\cdot S^\infty$, where $r\in \llbracket 0, N^*\rrbracket$ will be determined below.%
 \footnote{For every $k \geq 0$, we denote $\llbracket 0,k\rrbracket = \{0,1,\dots,k\}$.}
Players thus start with an experimentation phase that lasts $N^*$ periods. Player 1 experiments throughout this phase, while player 2 starts 
experimenting
with a delay of $r\geq 0$ periods. Once this phase is over, the continuation play hinges on whether one of the players was successful, as
under   $\sigma_0$. 
The value of $r$ is 
chosen to maximize  the expected payoff of player 2 
 over all $r\in \llbracket 0, N^*\rrbracket$. This completes the on-path description of $\sigma$.

 The off-path definition of $(\sigma,\pi)$ obeys the following principles.
 \begin{itemize}
 \item Any observable deviation from $\sigma$ is interpreted by the non-deviating player as evidence that the deviating player was successful, provided that the deviating player has experimented at least once prior to deviating, and induces the non-deviating player to repeat $R$ forever.\footnote{Following one's deviation, the deviating player best-replies to the continuation play of the non-deviating  player.} 
 \item The definition of $\sigma$ in the event where player 1 deviates at the initial node  depends on the expected payoffs $\gamma^1(\sigma)$ and $\gamma^2(\sigma)$ induced by $\sigma$. If  $\gamma^1(\sigma)>\delta \gamma^2(\sigma)$, 
  player 1 prefers to bear the cost of experimenting rather than to wait, and to exchange roles with player 2. In that case, players switch roles after $S$. That is,  we set $\sigma(Sh)= \sigma(h) $  for every $h \in H$,
so that player 1's  payoff conditional on playing $S$ is $\delta \gamma^2(\sigma)<\gamma^1(\sigma)$. 
  
  If instead $\gamma^1(\sigma)\leq \delta \gamma^2(\sigma)$,
player 1 would rather be in player 2's position, even with a one-stage delay. 
In this case, player 2 reacts to player 1's deviation by \emph{declining} to endorse the role of player 1: we set $\sigma(S)=S$ and $\sigma(SS\!\cdot\! h)= \sigma(h)$
for every $h \in H$.\footnote{In addition, $\sigma(SRh)= \sigma(Rh)$ for every $h \in H$.}
\item Deviations of player 2 that consist in advancing or postponing his first experiment are ignored.
\end{itemize}

\subsection{Reasonable sequential equilibria}

\label{sec refinement}
Under the sequential equilibrium $(\sigma,\pi)$, the choice by  player $i$ of the safe arm after an experiment is viewed by player $j$ as evidence that $i$'s experiment was successful. That is, player $j$'s off-path belief assigns probability 1 to player $i$ making a strictly dominated choice. 
Such beliefs are not ruled out by the concept of sequential equilibrium, yet they are excluded by leading refinements of this concept. 
Accordingly, we restrict ourselves to \emph{reasonable equilibria}, 
which we define next.\footnote{These are coined after 
\cite{mas-colell}, p.~468. } 
The relation of this concept to the existing equilibrium refinements is discussed in Section \ref{sec lit}.

\begin{definition}
A system $\pi$ of beliefs is \emph{reasonable} if the following holds for each  history $h'\in H$ with active player $i$:
\begin{itemize}
    \item If $h'$ is of the form $h'=  h S$, then $p^i(h')= \phi^{n_{\mathrm e}(h)}(p_0)$.
\item If $h'$ is of the form $h'= h Sa R$ with $a\in \{S,R\}$, 
then $p^i(h')= \phi^{n_{\mathrm e}(hSaR)}(p_0)$.
\end{itemize}
\end{definition}

The condition that $p^i(h S)= \phi^{n_{\mathrm e}(h)}(p_0)$ captures the fact that player $i$ infers from  $j$'s choice of $S$ at $h$ that  past experiments of player $j$ failed.

The additional condition on $p^i(h Sa R)$ captures the further requirement that, after $j$ has `revealed' by choosing $S$ that his experiments along $h$  failed, his choice of the risky arm in the next period is considered to be non-informative.

\begin{lemma}\label{lemm unique beliefs}
For every strategy  profile $\sigma$ there is a unique system of reasonable beliefs  consistent with $\sigma$.
\end{lemma}

\begin{proof}
Let $\sigma$ be given. We prove existence, then uniqueness.

\paragraph{Existence.} Let $(\sigma_n)_{n \in \dN}$ be a sequence of strategy profiles that converges to $\sigma$, with the property that the mixed action $\sigma_n(h)\in \Delta(\{S,R\})$ has  full support for each $h\in H$. According to $\sigma_n$, players always play  both arms with strictly positive probability unless  successful. Given $\sigma_n$, all public histories are on-path, hence the beliefs $\pi^i_n(\cdot)$ induced by $\sigma_n$ are uniquely defined and reasonable, and so any 
accumulation
point $\pi$ of $(\pi_n)_{n\in \dN}$ is reasonable as well. 
Note however that $\sigma_n$ is not completely mixed, since $\sigma_n$ repeats choosing $R$ once successful. 
Given $n$, pick next a sequence $(\tau_n^m)_{m \in \dN}$ of completely mixed profiles --- both arms are chosen with positive probability, even after a success --- such that $\lim_{m\to +\infty} \tau^m_n= \sigma_n$ for each $n$, in the product topology. Using a standard diagonal extraction argument, one can construct a sequence $(\tau_n^{m_n})_{n\in \dN}$ that converges to $\sigma$,  such that the belief system $(\pi_n^{m_n})_{n\in \dN}$ induced by $\tau_n^{m_n}$ converges to $\pi$.

\paragraph{Uniqueness.} We prove uniqueness by induction.\footnote{
We provide the proof for first-order beliefs $p^i(h)$. It extends easily to belief systems $\pi^i$, at the cost of extra notation.} Let $p^1(h)$ and $p^2(h)$ be given, and note that $p^1(h),p^2(h)>0$.\footnote{No finite amount of evidence can rule out $\theta= G$.} We show that $p^i(ha)$ is uniquely defined by $\sigma$ and the reasonable criterion, for each $i$ and $a\in \{R,S\}$. For concreteness, assume that player 1 is active at $h$. 
Since player~1 observes the outcome of his own experiments,
 one must have $p^1(hS)= p^1(h)$ and $p^1(hR)= \phi(p^1(h))$. 
On the other hand, 
 $p^2(hS)= \phi^{n_{\mathrm e}(h)}(p_0)$, since $(p^1,p^2)$ is reasonable. 
Consider finally $p^2(hR)$. 
If the last move of player 1 along $h$ is $S$, the reasonability criterion implies that $p^2(hR)= \phi^{n_{\mathrm e}(h)}(p_0)$. 
Assume  instead that the last move of player 1 is $R$. Since $\sigma$ dictates $R$ if this last experiment is a success, $p^2(hR)$ is uniquely deduced from $p^2(h)$ and $\sigma$ using Bayes rule.
\end{proof}

\begin{remark}\label{remark pure beliefs}
Let $\sigma$ be pure. The proof of Lemma \ref{lemm unique beliefs} shows that for each history $h$ and each choice $a\in \{S,R\}$, the reasonable beliefs associated with $\sigma$ are such that $p^i(ha)=1$ or $p^i(ha)\leq p^i(h)$. That is, along any sequence of choices, on- or off-path, the players become gradually more pessimistic until, possibly, they become convinced that the other player was successful. This pattern may repeat over time (off-path).

The proof also implies that for each history $h$ and each player $i$, one has $p^i(h)= 1$ or $p^i(h)= \phi^n(p_0)$, for some $n\geq n_{\mathrm e}^i(h)$, where $n_{\mathrm e}^i(h)$ is the number of experiments of $i$ along $h$.
\end{remark}
\begin{proposition}
\label{prop:reasonable:exist}
    Reasonable $\mathrm{SE}$ always exists.
\end{proposition}

\begin{proof}
    Consider a variant of the game in which a player has no longer access to the safe arm $S$ once successful. Any sequential equilibrium $(\sigma, \pi)$ of this variant is a \emph{reasonable} sequential equilibrium of the initial game.  
\end{proof}
\medskip

As we show in Section \ref{sec nopure},  \emph{pure} reasonable SE may  fail to exist.

\section{(Dis-)Encouragement effects and examples}\label{sec: encouragement}

\subsection{The encouragement cut-off}
An extra experiment of $i$ may lead to additional \emph{future} experiments of $j$. The marginal value for $i$ of this additional experiment by $j$ is ambiguous. 
Although it has an obvious positive externality on $i$ (the cost is borne by $j$, and the result is eventually known to $i$),  it may \emph{delay}   the point at which the outcomes of \emph{past} experiments of $j$ become known to $i$. There is no single cut-off that  captures all these trade-offs. 
We  discuss two `thought experiments' that lead to two (families of) cut-offs. 

\subsubsection{A first thought experiment}\label{sec first thought}
Let $h$ be a (public) history that ends with $S$. Assume that the active player, $i$, expects that choosing $R$ will induce $j$ to experiment once more, but that $j$ will pick the safe arm in case he, $i$, chooses $S$. Assume  moreover that later choices are $S$.\footnote{Except of course if one experiment, past or future, is successful.} 

Choosing $R$ yields a flow payoff of $(1-\delta)(p^i(h) \lambda m -c)$, and a per-stage continuation payoff of $g$ if one of the two experiments is successful. If $i$'s experiment is successful, this continuation payoff is received from the next period onwards;  if $j$'s experiment is successful but  $i$'s experiment is not, it is received with a delay of one period, since it is the fact that $j$ \emph{repeats} $R$ that will convey the good news to $i$. Consequently, $i$'s continuation payoff is  $ \delta p^i(h) \left( \lambda  + \delta (1-\lambda)\lambda \right)\times g$. 

Facing such a situation, $i$ prefers $R$ if and only if $p^i(h)\geq \widehat{p}$, where
\begin{equation}
\label{def:hatp}
\widehat{p}:=\frac{c(1-\delta)}{c(1-\delta)+g\bigl(1-\delta+\lambda \delta (1+\delta-\lambda \delta)\bigr)}<p^*.
\end{equation}
The cut-off $\widehat{p}$ reflects the positive externality of $j$'s experiment on $i$: 
if $p^i(h)\in (\widehat p,p^*)$, 
player $i$ is  willing to experiment further only if this induces $j$ to also experiment. 
The cut-off $\widehat p$, which plays a leading role in the analysis,
is called  the \textit{encouragement cut-off}.


There is a variant of the first thought experiment that incorporates the negative externality created when postponing the disclosure of information. In that variant, player $i$ is facing the same decision at $h$ as above, except that player $j$ has just experimented $n$ times in a row along $h$.\footnote{We also assume that the mere fact that player $j$ was experimenting is uninformative in itself, that is, that $j$'s strategy was dictating to experiment given previous choices. This is equivalent to assuming that $q^i(h)= 1-(1-\lambda)^n$.} 
Whether player $j$ was successful earlier will be inferred by player $i$, immediately if he chooses the safe arm, but with a one-period delay if he chooses the risky arm. 
This lowers the attractiveness of the risky arm relative to the first thought experiment, and choosing $R$ is optimal only if 
\[p^i(h) \geq  \widehat p_n:= \frac{c(1-\delta)}{c(1-\delta)+g((1-\delta)(1-\delta+\lambda \delta) + \delta(1- \delta(1-\lambda)^2)(1-\lambda)^n)}.\]

\medskip

\subsubsection{A second thought experiment}\label{sec second thought}

The prospect of learning the private information held by the other player lowers the option value of one's own experimentation. The alternative cut-offs we introduce below focus on this effect. Fix a history $h$, with active player $i$. Assume that player $j$ just experimented exactly $n$ times consecutively along $h$\footnote{As in the definition of $\widehat p_n$, we assume that these $n$ experiments are uninformative for $i$.} and is about to switch (forever) to $S$ if unsuccessful, independently of $i$'s choices. At $h$, the informational value of choosing $R$ is reduced for $i$, since this experiment is irrelevant in the event where $j$ was successful. Choosing $R$ at $h$ is optimal only if 
\begin{equation}\label{def:p*n} p^i(h) \geq  p^*_n:= \frac{ c(1-\delta)}{c(1-\delta) + g(1-\delta + \delta \lambda (1-\lambda)^n)}.\end{equation} \medskip

In both cases, the higher the stock $n$ of past experiments of $j$, the lower the attractiveness of choosing $R$: both sequences $(\widehat p_n)_{n\in \dN}$ and $(p^*_n)_{n\in \dN}$ are increasing with $n$, with $\widehat p_0=\widehat p$ and $p^*_0= p^*$.

\begin{remark}
In this second thought experiment,  the informational value of $i$'s experiment is negligible in the limit $n\to +\infty$, since $j$'s next choice will accurately reflect $\theta$. Thus, for very large $n$, $i$ is not willing to experiment unless the  net flow payoff from choosing $R$ is non-negative. That is, $\lim_{n\to \infty} p^*_n= p_{\mathrm{myop}}$, as letting $n\to +\infty$ in (\ref{def:p*n}) confirms.

By contrast, in the first thought experiment, the informational value of $i$'s experiment becomes \emph{negative} in the limit 
$n\to +\infty$. The reason is that choosing $R$ rather than $S$ will \emph{delay} the disclosure of the evidence held by $j$. Accordingly, $\lim_{n\to +\infty} \widehat p_n > p_{\mathrm{myop}}$. 
This implies that $\widehat p_n> p^*_n$ for large $n$, while $\widehat p_0 < p^*_0$.
\end{remark}

\begin{remark}
While $p^*_{n+1}>p^*_n$, 
simple algebraic manipulations show that
$\phi(p^*_{n+1})< p^*_n$. For $n=0$, the intuition  for this inequality is as follows. Consider the second thought experiment, and let  $p^i(h)< p^*_1$. At $h$, the optimal choice of $i$  is to choose $S$ and to learn the outcome of $j$'s most recent experiment. We claim that, if that experiment was not successful, the optimal choice of $i$ in the next period is to choose $S$. Otherwise, the optimal choice of $i$ in the next period would be $R$, irrespective of $j$'s experience. But if the optimal continuation payoff from choosing $S$ then $R$ after $h$ is non-negative, the optimal continuation payoff from choosing $R$ then $S$ would be even higher --- a contradiction.\footnote{This holds because changing the timing of the experiment does not affect $j$'s behavior.} Thus, $i$'s optimal choice at $h\cdot SS$ is $S$, which implies that $p^i(h\cdot SS)<p^*$. Since $p^i(h\cdot SS)= \phi(p^i(h))$, this implies that $\phi(p^*_1)<p^*$, as desired. 
\end{remark}

\subsection{Examples}\label{sec examples}
The structure of reasonable SEs is sensitive to the various parameters. 
We  provide three illustrations of this dependence. The first example  highlights the 
  role of $\widehat p$ in shaping the equilibrium level of experimentation. 
The second example  illustrates how to construct a pure reasonable SE for a class of parameter values.
  The last example illustrates why pure reasonable SEs may fail to exist due to a conflict between $p^*$ and $\widehat p$.

\subsubsection{The role of $\widehat p$}\label{sec: ex1}

We here show that the encouragement effect may lead players to experiment \emph{strictly} more when outcomes are private.
We assume $p_0\in (\max\{p^*, \phi^{-1}(\widehat p)\}, \phi^{-1}(p^*))$.
Since $\widehat p < p^*$, this interval is nonempty.
Since $\phi(p_0)<p^*< p_0$, there is exactly one experiment when outcomes are public.

We show that $N_{\mathrm e}\geq 2$ in any pure reasonable SE $(\sigma,\pi)$ of the game with private outcomes (assuming such an equilibrium exists).  We argue by contradiction, and assume that $N_{\mathrm e}=1$ when $\theta= B$, for some pure reasonable SE $(\sigma,\pi)$. It is w.l.o.g. to assume that $\sigma(\emptyset)=R$. 
Since $N_{\mathrm e}=1$, the equilibrium play is $RS\!\cdot\! S^\infty$, and hence $\sigma(RS)=S$: player 1 chooses $S$ in period 2 (if unsuccessful in period 1). 

Since on the other hand player 1 chooses $R$ in period 2 if successful in period 1 this implies that in equilibrium,
player 2 can deduce the outcome of player 1's  first experiment from player 1's choice in period 2. Hence  $p^2(RS\!\cdot\! R)=1$, implying $\sigma(RS\!\cdot\! R)=R$.

Therefore, the history $RS\!\cdot\! RR$ is on-path, but the history $\overline h= RS\!\cdot\! RR\!\cdot\! S$ is not.
At $\overline h$, player 1 has experimented twice, and has no information on the outcome of player 2's experiment, implying $p^1(\overline h)= \phi^2(p_0)$.
On the other hand,  $p^2(\overline h)=\phi^3(p_0)$ since beliefs are reasonable. As we show in the appendix (see Lemma~\ref{prop A2} in Section \ref{sec proof nopure}), the equilibrium path following $\overline h$ must be $S^\infty$.









Therefore, at $h= RS$, player 1 is facing a situation that is identical to the first thought experiment (Section \ref{sec first thought}).
Since $p^1(h)= \phi(p_0) > \widehat p$, choosing $R$ rather then $S= \sigma(RS)$ is optimal --- a contradiction.

\subsubsection{Constructing pure, reasonable $\mathrm{SE}$}\label{sec: ex2}

We define here a simple, pure strategy profile $\sigma$ which is a reasonable SE for some parameter values. According to  $\sigma$, the active  player experiments if no one has  experimented so far, or if the other player has just chosen to experiment and was the first one to experiment. Apart from such histories, the active player experiments if and only if he is convinced that the other player was successful. 
That is, for each history $h$ with active player $i$, we set  $\sigma(h)= R$ if either (i) $h=\emptyset$, (ii) $h= S^k\!\cdot\! a$ for some $k\geq 0$ and $a\in \{S,R\}$, or (iii) $p^i(h)=1$ holds.
There is no circularity in this definition, because $p^i(h)$ depends only on the definition of $\sigma$ at shorter histories.

Finding a characterization of the histories $h$ such that $\sigma(h)=R$ without using beliefs is tricky, because the interpretation by player $i$ of $j$'s previous choices hinges on how player $j$ previously interpreted earlier choices of $i$.
Figure 1 shows the equilibrium choices (highlighted) and the beliefs of the active player in the first periods. Below, we illustrate  the computation of beliefs at two representative nodes.

\begin{center}
\begin{tikzpicture}
\node[draw,cyan] (1) at (0,-8/3) {P1, $p_0$};
\node[draw,cyan] (3) at (4,0) {P1, $\phi(p_0)$};
\node[draw,cyan] (5a) at (8,-8/3) {P1, $\phi^2(p_0)$};
\node[draw,cyan] (5b) at (8,1/3) {P1, $1$};
\node[draw,cyan] (5c) at (8,8/3) {P1, $\phi^2(p_0)$};
\node[draw,cyan] (7a) at (14,-7/3) {P1, $\phi^3(p_0)$};
\node[draw,cyan] (7b) at (14,-2/3) {P1, $\phi^3(p_0)$};
\node[draw,cyan] (7c) at (14,1/3) {P1, $\phi^4(p_0)$};
\node[draw,cyan] (7d) at (14,7/3) {P1, $1$};
\node[draw,magenta] (2) at (2,-4/3) {P2, $p_0$};
\node[draw,magenta] (4a) at (6,-4/3) {P2, $\phi^2(p_0)$};
\node[draw,magenta] (4b) at (6,4/3) {P2, $1$};
\node[draw,magenta] (6a) at (11,-4/3) {P2, $\phi^2(p_0)$};
\node[draw,magenta] (6b) at (11,-1/3) {P2,$\phi^3(p_0)$};
\node[draw,magenta] (6c) at (11,4/3) {P2, $\phi^3(p_0)$};

\draw (0.8,-2) node[above]{$R$};
\draw (2.8,-2/3) node[above]{$R$};
\draw (4.8,2/3) node[above]{$R$};
\draw (4.7,-2/3) node[below]{$S$};
\draw (6.7,-2) node[below]{$S$};
\draw (6.7,2) node[above]{$R$};
\draw (6.7,-0.6) node[above]{$R$};
\draw (9.3,-0.6) node[above]{$S$};
\draw (9.3,-2) node[above]{$R$};
\draw (9.3,0.8) node[above]{$R$};
\draw (12.5,-1.9) node[below]{$S$};
\draw (12.5,1.9) node[above]{$R$};
\draw (12.5,-0.4) node[below]{$R$};
\draw (12.5,0.8) node[below]{$S$};

\draw[thick](1) -- (2);
\draw[thick](2) -- (3);
\draw[thick](3) -- (4a);
\draw[thick](4a) -- (5a);
\draw[thick](4b) -- (5c);
\draw[thick](5b) -- (6c);
\draw[thick](6a) -- (7a);
\draw[thick](6c) -- (7c);

\draw[color=gray!40] (3) -- (4b);
\draw[color=gray!40] (4a) -- (5b);
\draw[color=gray!40] (5b) -- (6b);
\draw[color=gray!40] (5a) -- (6a);
\draw[color=gray!40] (6a) -- (7b);
\draw[color=gray!40] (6c) -- (7d);
\end{tikzpicture}

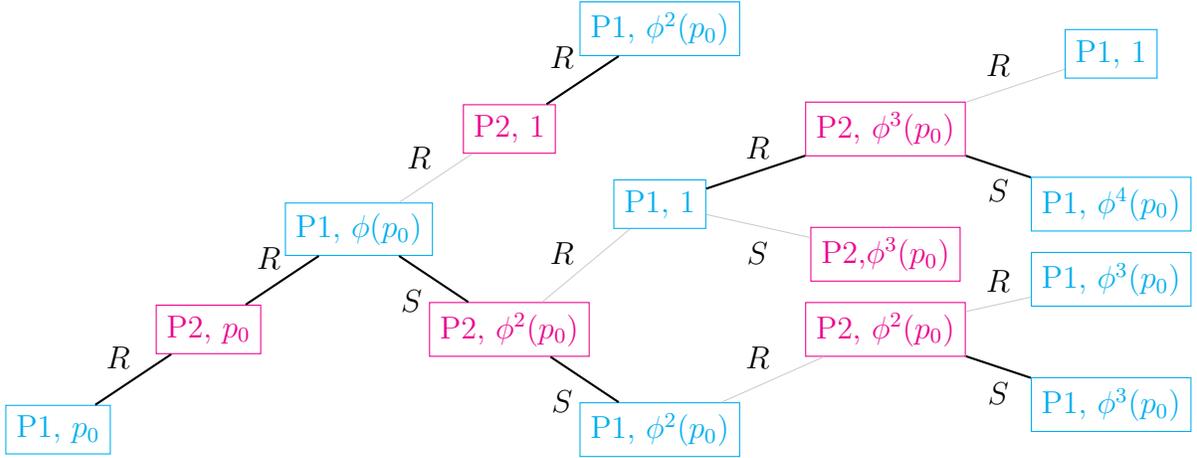
Figure 1: Choices and beliefs.
\end{center}

\smallskip\noindent\textit{Node 1:} $h= RR\!\cdot\! R$.
Since 
$\sigma(RR)= S$, the (on-path) history $h$ occurs if and only if player 1's first experiment is a success. Hence, $p^2(h)=1$ and $\sigma(h)=R$.

\smallskip\noindent\textit{Node 2:  $h= RS\!\cdot\! SR\!\cdot\! RR$. } This case is more involved.
Along $h$, player 1's choice of $S$ in period 2 is evidence that player 1's experiment in period 1 failed, hence player 1's experiment in period 3 following $RS\cdot SR$ is uninformative, so that $p^2(h)= \phi^3(p_0)$. 
By contrast, while player 2's choice of $R$ in period 2 is uninformative,  the \emph{second} experiment of player 2  in period 3 is evidence that the first was a success, since player 2 would have played $S$ otherwise (because  $\sigma(RS\!\cdot\! SR\!\cdot\! R)= S$). Hence,  $p^1(h)=1$ and $\sigma(h)=R$.

\begin{proposition}\label{ex pure reasonable SE}  If $p_0\in [p^*_1, \phi^{-1}(\widehat{p})]$, then  $\sigma$ is a   reasonable $\mathrm{SE}$.
\end{proposition}

For 
$p_0\in [p^*_1, \phi^{-1}(\widehat{p})]$,
we have
$\phi(p_0)\leq \widehat p< p^*$. 
Hence $N^*=1$, while $N_{\mathrm{e}}=2$. We note that $[p^*_1, \phi^{-1}(\widehat{p})]$ may be empty or not, depending on parameter values.

\medskip

The proof of Proposition \ref{ex pure reasonable SE} is in Section \ref{sec proof of pure reasonable SE}. 
We observe here  that the  inequality $p_0\geq p^*_1$ ensures that player 2 finds it optimal to choose $\sigma(R)=R$ after $R$. The other restriction  $p_0<\phi^{-1}(\widehat p)$,  implies that $p^1(RS)<\widehat p$, so that  player 1 does not benefit from pretending that his first experiment was successful, even if that would lead player 2 to experiment once more. 

\begin{remark}
\label{remark:6}
We also construct in the supplementary material a pure reasonable $\mathrm{SE}$  for the case where $p_0\in [p^*,\min(p^*_1,\phi^{-1}(\widehat p))]$. In that equilibrium, a player experiments only if no one has experimented  in the past, or if convinced that the other was successful.
\end{remark}


\subsubsection{On the non-existence of pure reasonable SEs}\label{sec nopure}

Unlike typical  signalling games,  our  strategic experimentation games always have a pure SE (Section \ref{sec pure SE}). 
Pure reasonable SEs need not exist.




\begin{proposition}\label{ex4}
If $\phi^{-1}(\widehat p)<p^*$, there is no pure reasonable $\mathrm{SE}$ for 
$p_0\in (\phi^{-1}(\widehat p),p^*_1)$.
\end{proposition}

Proposition \ref{ex4}, together with Remark~\ref{remark:6}, clarifies the existence issue of a pure reasonable SE for $p_0$ close to $p^*$. In particular, for $p_0=p^*$, there is a pure reasonable SE if and only if $\phi^{-1}(\widehat p)\geq p^*$.
When this inequality does not hold, players must resort to randomization. The nature of such mixed equilibria is illustrated in the supplementary material.

\medskip
\begin{proof}[Proof Sketch]
Assume that a pure reasonable SE $(\sigma,\pi)$ exists. As we show in Lemmas \ref{prop A2} and \ref{lemm last2} (Section \ref{sec proof nopure}), the assumption that $p_0<p^*_1$ implies that the equilibrium sequence of choices is  $S^\infty$, following both histories $RS$ and $RS \!\cdot\! RR\!\cdot\! S$. 
Therefore $\sigma(RS)= S$, and hence $p^2(RS\!\cdot\! R)= 1$. 
The second property implies that at $h= RS$, player 1 is in the situation of the first thought experiment of Section \ref{sec first results}.  Since $p^1(h)= \phi(p_0)>\widehat p$, deviating to $R$ (and then $S$) yields a higher continuation payoff than the equilibrium choice $S$. This is the desired contradiction.
\end{proof}

\section{Main results}\label{sec main results}

\subsection{Statements}

According to Theorem~\ref{prop private}, players experiment more when outcomes are private than when  they are public. Our main results provide additional insights into the equilibrium amount of experimentation. 

One way of addressing this question is to describe the final beliefs of the players in terms of cut-offs such as $p^*$, $p^{**}$, or $\widehat p$. This is the approach we follow in Theorems \ref{th private1} and \ref{th private2}  below.

Unlike in the public case, the relation between one's belief and the number of past experiments is endogenous, as it depends on how to interpret the other player's experiments. 
Hence, an  alternative approach 
to measuring the equilibrium amount of experimentation is to compare $N_{\mathrm e}$ to $N^*$ and $N^{**}$. This is the focus of Theorem \ref{th private3}.


\begin{theorem}\label{th private1}
    For any reasonable SE $(\sigma,\pi)$, the following holds:
    \begin{itemize}
        \item If $p_0>p^*$, final beliefs are below $\widehat p$  with $\prob_\sigma$-positive probability.
        \item For every $h$ s.t. the belief of the active player satisfies $p^i(h)<\widehat p$, one has $\sigma(R\mid h) = 0$.
    \end{itemize}
\end{theorem}
In equilibrium, either some player is successful, in which case beliefs converge to 1, or both players remain unsuccessful, in which case experimentation eventually stops. According to Theorem \ref{th private1}, the beliefs are then below $\widehat p$ with positive probability. Final beliefs need not be below $\widehat p$ with probability 1 (see supplementary appendix).

Theorem \ref{th private1} also implies that as soon as the active player's belief drops below $\widehat p$, experimentation stops forever,
except if the other player, $j$, was successful.  
Indeed, the active player,  $i$, must choose $S$ at $h$. Since beliefs are reasonable, 
$p^j(hS)= \phi^{n_{\mathrm e}(h)}(p_0)\leq p^i(h) <\widehat p$. Applying Theorem \ref{th private1} again, this implies $\sigma(hS)=S$, etc.

\medskip
    
For \emph{pure} reasonable SEs, the contrast with the public case is even more explicit.

\begin{corollary}
Assume that outcomes are \emph{private} and $p_0> p^*$. At any \emph{pure} reasonable $\mathrm{SE}$, final beliefs  are  below $\widehat p$ if $\theta= B$, and no  player ever experiments once his belief is below $\widehat p$.   

Assume that outcomes are \emph{public}. At any Nash equilibrium, final beliefs  are  below $p^*$ if $\theta= B$, and no  player ever experiments once his belief is below $p^*$. 
\end{corollary}

\medskip

By Remark \ref{remark1},
the strategy profile that experiments whenever the current belief is above $p^*$ is an SPE in the public case for a continuum of values of $p_0>p^*$, if and only if $\phi^{-1}(p^*)>p_{\mathrm{myop}}$. 
It is thus natural to ask whether the  profile in which the active player experiments whenever his current belief is above $\widehat p$ is a reasonable SE in the private case. 
Theorem \ref{th private2} below provides the answer.

\begin{theorem}
\label{th private2} 
Let  $\sigma$ be the 
pure 
strategy profile such that $\sigma(h)= R$ if and only if the belief of the active player satisfies 
 $p^i(h)\geq \widehat p$.
Then $\sigma$ is a reasonable $\mathrm{SE}$ if and only if $\phi^{n-1}(p_0)\geq p^*_n$,
where $n:=\min\{k\geq 0: \phi^k(p_0)< \widehat p\}$.
\end{theorem}

The condition $\phi^{n-1}(p_0)\geq p^*_n$ is restrictive. The example in Section~\ref{sec: ex2} corresponds to the case $n=1$.

Theorem \ref{th private3} below provides bounds on the equilibrium amount of experimentation $N_{\mathrm e}$, in terms of the socially optimal amount $N^{**}$.

\begin{theorem}\label{th private3}
    Assume $p_0> p^*$. At any pure reasonable SE $(\sigma,\pi)$, if $\theta= B$, the equilibrium number of experiments satisfies 
    \[N^{**}-2\leq N_{\mathrm e} \leq 2 N^{**}.\]
Moreover,  for every $\alpha<2$ there are parameter values such that there is a pure reasonable $\mathrm{SE}$ satisfying $N_{\mathrm e}\geq \alpha N^{**}$.
\end{theorem}

According to Theorem \ref{th private3},  the equilibrium amount of experimentation is always at least socially optimal (up to the additive constant  2) in any pure reasonable SE, in  contrast to the public case. 
On the other hand, over-experimentation may occur,
and Theorem \ref{th private3} provides a tight bound on the maximal amount of over-experimentation.\footnote{The factor 2 relates to the number of players: each player experiments at most  $N^{**}$ periods.}

According to Proposition~\ref{prop comparison p/e}, in \emph{Nash} equilibria 
the ratio $N_{\mathrm e}/N^{**}$ can be arbitrarily close to 2.299.
Theorem \ref{th private3} shows that for reasonable SEs, the ratio does not reach this quantity.

Part of the proof of Theorem \ref{th private3} relies on the observation that the two cut-offs $\widehat p$ and $p^{**}$ are closely related. We state this observation as an independent lemma.

\begin{lemma}\label{lemm last}
 One has $\phi^2(\widehat p)< p^{**}< \widehat p$.
\end{lemma}

\subsection{Proof sketches}

\begin{proof}[Proof sketch of Theorem \ref{th private1}] The assumption that $p_0> p^*$ ensures that $N_{\mathrm{e}}\geq 1$, with probability 1.
If final beliefs are always above $\widehat p$, then the number of experiments $N_{\mathrm e}$ is bounded on-path. Hence,  there must exist a `last experiment', that is, some on-path history $h$ that ends with $R$ and  is followed by $S^\infty$ with probability 1. In that case, the player active at $h\!\cdot\! S$ has experimented in the previous period, and is in the situation of the first thought experiment. Since his belief is above $\widehat p$, deviating to the risky arm is profitable. 

The proof of the second statement is more involved. We argue by contradiction and
assume that  $p^i(h)<\widehat p$, yet $\sigma(R\mid h)>0$ for some history $h$ with active player $i$.

Let $\omega_i:=\inf\{k\geq 1: \sigma(S\mid h\cdot (RR)^k)>0\}$ 
be the first period beyond the current one in which player $i$ chooses $S$ with \emph{positive} probability. Let $\omega_j:=\inf\{k\geq 0: \sigma(S\mid  h\! \cdot\! (RR)^k\!\cdot\! R)= 1\}$ be the first period in which player $j$ chooses $S$ with probability \emph{one}.

It cannot be that $\omega_i= \omega_j= +\infty$. Indeed, if $\omega_i=+\infty$,  player $i$  experiments with probability one as long as player $j$ experiments,  hence $i$'s experiments after $h$ are uninformative from $j$'s perspective. This implies that the belief $p^j(h\!\cdot\! (RR)^k\!\cdot\! R)$ converges to 0 as $k\to +\infty$, implying $p^j(h\!\cdot\! (RR)^k\!\cdot\! R)< \widetilde p$ for $k$ large, 
and hence $\omega_j<+\infty$. 
\smallskip

For simplicity, we assume in this sketch that $h$ ends with $S$ and that  $\sigma$ induces the continuation path $S^\infty$ after any extension of $h$ ending with $S$. That is, $\sigma(h\overline hS)=S$ for each $\overline h$. 
We also assume that $\omega_j< \omega_i$. 
These assumptions are not w.l.o.g., 
and the complete proof 
in Section~\ref{sec proof private1} involves a few additional complications. 
We place ourselves in the (continuation) game that starts at $h$, relabeling periods starting from 0 at $h$, and prove that the equilibrium (continuation) payoff of $i$ is negative --- a contradiction.

Under the assumption that $\omega_j<\omega_i$, 
player $i$ chooses $R$ until the first (random) time where $j$ chooses $S$, which occurs in period $\omega_j$ at the latest. 
Unlike in the rest of the paper, it will be convenient to index strategies with players and write $\sigma= (\sigma_i,\sigma_j)$. For $k\leq \omega_j$, we denote by $\sigma_j(k)$ the (continuation) strategy of player $j$ that chooses $R$ $k$ times in a row and then  $S$,  
and otherwise coincide with $\sigma_j$. 

Against $\sigma_i$, the \emph{behavior} strategy $\sigma_j$, which switches to $S$ in period $k$ with probability%
\footnote{Conditional on not being successful and on not switching to $S$ before.}
$\sigma(S\mid h\!\cdot\! (RR)^{k}\!\cdot\! R)$ is equivalent to some \emph{mixture} of the strategies $\sigma_j(k)$, $k\in \llbracket 0, \omega_j\rrbracket$. In particular, the expected (continuation) payoff of $i$ is a convex combination of the expected payoffs induced by the profiles $(\sigma_i,\sigma_j(k))$. 

It is sufficient to prove that for \emph{each} $k$, the continuation payoff of player $i$ given $(\sigma_i,\sigma_j(k))$ is negative. That is, even if $i$ knew in advance when $j$ would stop experimenting, $i$ would be better off choosing $S$ at $h$.

Fix $k$, and assume $i$ `knows' he is facing $\sigma_j(k)$. Under $(\sigma_i,\sigma_j(k))$, the players experiment for $k$ periods, then $i$ experiments one last time. The subsequent sequence of choices indicates whether players were successful or not. When facing $\sigma_j(k)$, 
player $i$'s belief after the first $k$ periods is $\phi^{k}(p^i(h))<\widehat p< p^*_{k}$.
Hence,
$i$'s payoff is strictly lower under $(\sigma_i,\sigma_j(k))$ than it would be if $i$ would instead choose $S$ in period $k$, rather than experimenting one last time. 
Note that if $k>0$, the payoff of player $i$ under this alternative strategy coincides with his payoff in the first thought experiment in Section \ref{sec: encouragement} and is therefore negative, since $p^1(h)< \widehat p$.\footnote{If $k=0$, the alternative strategy chooses $S$ at $h$, and yields a payoff of zero.}
\end{proof}

\medskip

\begin{proof}[Proof sketch of Theorem \ref{th private2}]
Consider the public history $h=(RR)^{n-1}\!\cdot\! R$, at which player 2 is active and holds the belief
 $p^2(h)= \phi^{n-1}(p_0)$. 
When at $h$, player 2 anticipates that he is about to learn whether player 1 was successful in the first $n$ periods and is thus in the situation of the second thought experiment in Section \ref{sec: encouragement}. Thus, deviating to $S$ is optimal unless $p^2(h) \geq p^*_n$. This proves the necessity part.
The  proof of the sufficiency part is much more involved and appears in Section~\ref{appendix:thm2022-1}. 
It requires to prove that, for each history, the payoff on the equilibrium continuation path exceeds the payoff on the continuation path induced after a one-step deviation. 
\end{proof}
\medskip

\begin{proof}[Proof of Lemma \ref{lemm last}]
Suppose that the  prior is such that the expected payoff of player 1 is zero in the following scenario:  both players experiment 
in period 1,
then the experiment outcomes are 
publicly
disclosed, and no further experimentation takes place
from period 2 and on (unless one of the experiments in period 1 was successful).
Denote by $\overline p$ the corresponding value of $p_0$.

In this scenario, the overall payoff of player 2 is also zero, hence the social planner's payoff is zero, implying that $\overline p\geq p^{**}$. On the other hand, this scenario is more favorable to player 1 than the first thought experiment, since a possible success of player 
2
in period 1 is immediately made public. Hence, $\overline p< \widehat p$. This implies $p^{**} <\widehat p$. 

    Observe next that if the prior is such that $p_0>  \phi^{-1}(p^{**})$, a social planner would choose to experiment at least twice. Hence the expected payoff of player 1 in the above scenario is positive, implying $\phi^{-1}(p^{**})\geq \overline p$.

    To complete the proof, it remains to show that $\widehat p<\phi^{-1}(\overline p)$. Our argument is algebraic. The cut-off value $\overline p$ is given by $\overline p= \displaystyle \frac{c(1-\delta)}{c(1-\delta)+ g(1-\delta(1-\lambda)^2)}$, so that 
    $\displaystyle \phi^{-1}(\overline p )=\frac{c(1-\delta)}{c(1-\delta)+ g(1-\lambda)(1-\delta(1-\lambda)^2)}$. Elementary manipulations show that the inequality $\widehat p<\phi^{-1}(\overline p)$ is equivalent to the inequality
    \[\lambda(1-\delta)^2+\lambda^2\delta(3-\delta-\lambda)>0,\]
   which trivially holds.
\end{proof}

\medskip\begin{proof}[Proof sketch of Theorem \ref{th private3}]
Set $\widehat N:=\inf\{n\geq 1 \colon \phi^n(p_0)< \widehat p\}$. By Lemma \ref{lemm last}, one has $N^{**}-2\leq \widehat N\leq N^{**}$.
Let  $(\sigma,\pi)$ be a pure reasonable SE. By Theorem \ref{th private1},  the final belief, $\phi^{N_{\mathrm e}}(p_0)$, is at most $\widehat{p}$ in the absence of a success, which implies%
\footnote{Plainly, $N_{\mathrm e}=+\infty$ if a player is successful.}
$N_{\mathrm e}\geq \widehat{N} \geq N^{**}-2$.

Suppose now that $\theta=B$. For any on-path history $h$ with active player $i$ one has $p^i(h) \leq \phi^{n^i_{\mathrm e}(h)}(p_0)$, see Remark \ref{remark pure beliefs}. This implies that $p^i(h)< \widehat p$ as soon as $n^i_{\mathrm e}(h)\geq \widehat N$. The second part of Theorem \ref{th private1} implies that $\sigma(h) = S$. Hence the total number of experiments is at most $N_{\mathrm{e}}\leq 2 \widehat N\leq 2N^{**}$.

The proof of the second part of  Theorem \ref{th private3} is based on Theorem \ref{th private2}.
\end{proof}
\section{Related literature}\label{sec lit}

The literature on strategic experimentation with exponential bandits is by now quite large, see, e.g., \cite{das2020strategic}, \cite{keller2010strategic}, \cite{keller2015breakdowns}, \cite{klein2011negatively}, \cite{marlats2021strategic}.
However, 
there are only a few papers  that relax the assumption that both actions/efforts and outcomes are observed. \cite{bonatti2011collaborating} focus on the hidden actions/observed payoffs case, and develop a model where agents collectively work on a project with uncertain prospects, facing a moral hazard problem due to the inability to monitor each other's actions. 
They find that
free-riding behavior among team members leads to reduced effort and procrastination,
and that simply improving monitoring does not necessarily lead to better outcomes.

\cite{rosenberg2007social} analyze a discrete-time, two-arm bandit problem with hidden outcomes, but assume that the choice of the safe arm is irreversible. \cite{murto2011learning} analyze information aggregation in a  game with irreversible exit, with private payoffs and public actions. 

\cite{heidhues2015strategic} consider a model that is essentially identical to ours, but allow for cheap talk communication between the players. They show that the socially optimal level of experimentation can be achieved when initial beliefs are sufficiently optimistic. 
They also find pure sequential equilibria that exhibit over-experimentation. However, their sequential equilibria rely on `non-reasonable' beliefs, which we rule out here.

The issue of over-experimentation relative to the socially efficient level also arises in 
\cite{halac2017contests}, who examine  how to design dynamic contests that promote innovation when there is uncertainty about the feasibility of the innovation, and show the optimality of hiding information about successes, under certain conditions.
\smallskip

The idea that some sequential equilibria are  unreasonable because they rely on incredible beliefs has led to many refinements, 
see the surveys \cite{hillas2002foundations} and \cite{van2002strategic}.
While most concepts were defined for finite extensive games and/or normal form games, 
and hence do not apply as such, the idea that underlies our notion of reasonable beliefs is implied by most existing refinements. For instance, any quasi-perfect equilibrium 
(\cite{van1984relation}) corresponds to a reasonable equilibrium.\footnote{Formally, this applies to finite-horizon truncations of our game.} A fortiori, any proper equilibrium 
(\cite{myerson1978refinements})
of the normal form of the (truncation) of the experimentation game corresponds to a reasonable equilibrium in the extensive form. The restriction to reasonable beliefs is also reminiscent of refinements defined for signaling games (\cite{banks1987equilibrium}, \cite{cho1987signaling}), and is implied by the \emph{intuitive criterion} as defined in \cite{cho1987refinement}.


\bibliographystyle{plainnat}
\bibliography{bibliography}

\begin{appendix}
%

\section{Proof of Theorem \ref{prop private}}
\label{section:proof:prop private}

Fix a Nash equilibrium $\sigma$, and denote by $N_{\mathrm s}\in \dN\cup \{+\infty\}$ the total number of times $S$ is used over time. We first prove that either $N_{\mathrm s}<+\infty$, or $N_{\mathrm e}<+\infty$: players eventually settle on the same arm. 

\begin{claim}\label{cl 1} One has $\prob_{\sigma}(N_{\mathrm s}=+\infty\mbox{ \textrm{and} } N_{\mathrm e}=+\infty)=0$.
\end{claim}

\begin{proof}
Assume there is an on-path history of the form $h= \overline hS$  such that $n_{\mathrm e}(h)>\widetilde N$, and let $i$ be the player active at $h$. Then
$p^i (h)= \phi^{n_{\mathrm e}(h)}(p_0)< \phi^{\widetilde N}(p_0) \leq \widetilde p$, and hence  $\sigma(h)=S$. By the same argument,  $\sigma(hS^k)=S$ for each $k\geq 1$.
 This implies that on the event $N_{\mathrm e}> \widetilde N$, one has either $N_{\mathrm s}<+\infty$, or $N_{\mathrm e}<+\infty$, $\prob_\sigma$-a.s. The result follows.
\end{proof}

\medskip

We next prove that $\prob_{\sigma}(hR^\infty\mid B)=0$, for every history $h\in H$. 
Summing over the countable set $H$ of finite histories will imply that $\prob_{\sigma}\left(N_{\mathrm s}<+\infty\mbox{ and } N_{\mathrm e}=+\infty\mid B\right)=0$. Combined with Claim \ref{cl 1}, this implies that $N_{\mathrm e}<\infty$ a.s.~when $\theta=B$.

\begin{claim}\label{cl 2}
One has  $\prob_{\sigma}(hR^\infty\mid B)=0$ for each $h\in H$.
\end{claim}

\begin{proof}
We argue by contradiction, and let $h\in H$ be such that  $\prob_{\sigma}(hR^\infty\mid B)>0$. For $n\geq 0$,   let  $h_n$ denote the restriction of $hR^\infty$ to the first  $n$ periods, so that $|h_n|=2n$. 
Denote by $h^1_n$ the \emph{private} history of player 1 along which the sequence of choices is as in $h_n$ and all experiments of player 1 failed. 
Then $p^1(h_n)= \prob_\sigma(\theta= G\mid h^1_n)$.

By Bayes rule, 
\begin{equation}\label{finite number}
\frac{p^1(h_n)}{1-p^1(h_n)}= \frac{\prob_\sigma\left(\theta= G\mid h^1_n\right)}{\prob_\sigma\left(\theta= B\mid h^1_n\right)}= 
\frac{p_0}{1-p_0}\times \frac{\prob_{\sigma}\left(h^1_n\mid \theta= G\right)}{\prob_{\sigma}\left(h^1_n\mid \theta = B\right)}.\end{equation}

Plainly, $\prob_{\sigma}\left(h^1_n\mid \theta = B\right)= \prob_{\sigma}(h_n\mid B)$, and therefore converges to $\prob_{\sigma}(hR^\infty\mid B)>0$. 
Hence the denominator on the RHS of (\ref{finite number}) has a non-zero limit.

Since the probability of a failure is $1-\lambda$ if $ \theta= G$, we have 
$\prob_{\sigma}\left(h^1_n\mid \theta = G\right)\leq (1-\lambda)^{n_e^1(h_n)}$. Since $n_e^1(h_n)\to +\infty$, Eq.~\eqref{finite number}
implies that $\lim_{n\to +\infty}p^1(h_n)=0$. In particular, $p^1(h_n)<\widetilde p$ for all $n$ large enough, so that $\sigma(h_n)=S$ for all such $n$'s. 
This contradicts the assumption that $\prob_{\sigma}(hR^\infty\mid B)>0$.
\end{proof}

\section{The proofs for Section \ref{sec examples}: examples}

\subsection{Proof of Proposition \ref{ex pure reasonable SE}}\label{sec proof of pure reasonable SE}

Let $h$ be an arbitrary history, with active player $i$. We prove that $i$ has no profitable one-step deviation at $h$. If $i$ is convinced that $j$ was successful, that is, if $p^i(h)=1$, then choosing $R= \sigma(h)$ is optimal. We henceforth assume that $p^i(h)<1$.

Note that if $h= \overline h \!\cdot\! RSR$ ends with $RSR$, then  the definition of $\sigma$ in Section \ref{sec: ex2} implies that $\sigma(\overline h\!\cdot\! RS)= S$, and hence $p^i(h)=1$. 
The assumption $p^i(h)<1$ thus rules out such histories.

\smallskip\noindent
\textbf{Case 1.} Assume that $h= S^k$ for some $k\geq 0$. The continuation play induced by $\sigma$ is $RR\cdot S^\infty$. If instead $i$ deviates to $S$, the continuation play is $SR\cdot RS\cdot S^\infty$. Relabeling periods starting from $h$,  player $i$ experiments once in both cases, and observes the outcome of $j$ in period 2, after his own experiment. Hence which sequence is preferred by $i$ depends on whether he prefers to choose $R$ in period 1 or in period 2. Since $p_0\geq p^*$, the former 
is preferred.

More formally, the continuation payoff induced by $\sigma$ is 
\[-c(1-\delta)+p_0\lambda( (1-\delta)m+\delta g) +p_0 \lambda (1-\lambda) \delta^2g,\] while the payoff when deviating is $\delta(-c(1-\delta)+p_0\lambda((1-\delta)m+\delta g)+p_0(1-\lambda)\lambda\delta^2 g$. 
The former is preferred if and only if 
$-c(1-\delta)+p_0\lambda( (1-\delta)m+\delta g)\geq 0$, that is, $p_0\geq p^*$.

\smallskip\noindent
\textbf{Case 2.} Assume that $h= S^kR$ for some $k\geq 0$. Choosing $R= \sigma(h)$ induces the continuation play $RS^\infty$. Deviating to $S$ induces $S^\infty$. Since $p^i(h)= p_0\geq p^*_1$, the deviation is not profitable.

\smallskip\noindent
\textbf{Case 3.} Assume that $h=\overline h\!\cdot\! SS$ for some $\overline h$, with $n_{\mathrm e}(h)\geq 1$. Choosing $S= \sigma(h)$ induces the continuation play $S^\infty$. Deviating to $R$ induces $S^\infty$, because $p^j(hR)<1$. Since $p^i(h)= \phi^{n_{\mathrm e}(h)}(p_0)\leq \phi(p_0)\leq \widehat p\leq p^*$, the deviation is not profitable.

\smallskip\noindent
\textbf{Case 4.} Assume that $h= \overline h\cdot RS$  for some $\overline h$. Choosing $S=\sigma(h)$ induces the continuation play $S^\infty$. If $i$ chooses instead $R$, the belief of $j$ will jump to $p^j(hR)=1$, hence $j$ will choose $R$. According to $\sigma$, $i$ will follow with $S$, so that the continuation play is $RR\cdot S^\infty$. Since $p^i(h)= \phi^{n_{\mathrm e}(h)}(p_0)\leq \widehat p$, the deviation is not profitable.

\smallskip\noindent
\textbf{Case 5.} Assume that $h= \overline h\!\cdot\! RR$ for some $\overline h$. Since $p^i(h)<1$, choosing $S= \sigma(R)$ induces the continuation play $S^\infty$. As in \textbf{Case 4}, the continuation play  is $RR\cdot S^\infty$ if $i$ deviates to $R$. Since $p^i(h)<1$, the belief of $i$ is equal to $\phi^n(p_0)$ for some 
$n\geq 1$, hence $p^i(h)\leq \widehat p$ and the deviation is therefore not profitable.

\smallskip\noindent
\textbf{Case 6.} Assume that $h= \overline h\!\cdot\! SSR$ for some $\overline h$, with $n_{\mathrm e}(h)\geq 2$. 
Choosing $S= \sigma(R)$ induces the continuation play $S^\infty$, while deviating to $R$ induces the continuation play $RS^\infty$. From the structure of $h$, it follows that $p^i(h)= \phi^{n_{\mathrm e}(\overline h)}(p_0)\leq p^*_1$, hence the deviation is not profitable.

\subsection{Proofs for Section \ref{sec nopure}}\label{sec proof nopure}

\begin{lemma}\label{prop A1}
Let $\sigma$ be a pure reasonable $\mathrm{SE}$.
If $p_0<p^*$, then the 
path induced by $\sigma$ is $S^\infty$.
\end{lemma}

\begin{proof} We first show that $\sigma(hS)=S$ whenever $p^i(hS)<1$ for each $i$. We argue by contradiction and assume that the set
\[H_0:=\{h: p^1(hS)<1,\ p^2(hS)<1, \mbox{ and } \sigma(hS)=R\}\] 
is non-empty.

Let $h\in H_0$ be arbitrary and $i$ be the player active at $hS$. 
Since $\sigma$ is a pure reasonable SE, one has $p^i(hS)= \phi^{n_{\mathrm e}(h)}(p_0)$. On the other hand, $p^i(h)\geq \widetilde p$, since $\sigma(hS)=R$. 
Therefore, $\max_{h \in H_0}n_{\mathrm e}(h)\leq \widetilde N$.

Let $h \in H_0$ be such that
$n_{\mathrm e}(h)= \max_{h' \in H_0} n_{\mathrm e}(h')$. We prove that player $i$ has a profitable one-step deviation at $h$, which will yield the desired contradiction.

Since $h$ maximizes $n_{\mathrm e}(\cdot)$ over $H_0$, the  continuation play induced by $\sigma$ after $hS$ is either equal to  $R^kS^\infty$ for some $k\geq 1$, or to $R^\infty$. We rule out both cases in turn.

\smallskip\noindent
\textbf{Case 1:} The continuation play is $R^\infty$. 

In this case,  $p^i(hS\!\cdot\! (RR)^k)=\phi^{n_{\mathrm e}(h)+k}(p_0)$ for each $k$, and hence  $p^i(hS\!\cdot\! (RR)^k)<\widetilde p$  for all $k$ large enough, contradicting the assumption $\sigma(hS\!\cdot\! (RR)^k)=R$. 

\smallskip\noindent
\textbf{Case 2:} The continuation play is $ R^k S^\infty$, for some $k\geq 1$.
    
Assume first that $k=1$, and let us place ourselves at the history $hS$. Since player $j$ just played $S$, player $i$ expects his experiment $R=\sigma(h)$ at $hS$ to be the last.
Since $p^i( hS)\leq p_0<p^*$, the equilibrium continuation payoff of $i$ at $hS$ is (strictly) negative. 
On the other hand the optimal continuation payoff when deviating to $S$ is non-negative, irrespective of the continuation strategy of player $j$ --- a contradiction.

Assume now that $k>1$.
Set $\widehat h:= hS\cdot R^{k-1}$, and let $i$ be the player active at $\widehat h$.
Since $p^1(hS),p^2(hS)\leq p_0$, and since $\sigma$ induces the sequence of choices $R^{k-1}$ after $hS$, 
it follows that 
$p^i(\widehat h)\leq p^i(hS)\leq p_0<p^*$.
By the maximality property of $ h$ and since $n_{\mathrm e}(\widehat h)> n_{\mathrm e}( h)$, the continuation play 
induced by $\sigma$ after $\widehat h a$ is $S^\infty$ for each $a\in \{R,S\}$. Since $p^i(\widehat h)<p^*$, 
player $i$'s continuation payoff at $\widehat h$ is higher when deviating to $S$ than with $R=\sigma(\widehat h)$. 
\medskip

We now conclude by showing that $N_{\mathrm e}=0$ under $\sigma$. Assume instead that $N_{\mathrm e}\geq 1$. 
The first part of the proof implies  that the play  induced by $\sigma$ is of the form $R^kS^\infty$ for some $k\geq 1$. A contradiction is obtained by repeating the 
arguments in Case 2. 
This concludes the proof of Lemma \ref{prop A1}.
\end{proof}

\begin{lemma}\label{prop A2}
Let $\sigma$ be a pure reasonable $\mathrm{SE}$. Let $h$ be a history such that $p^i(h)<p^*$ for each $i$. Then the continuation play following $h$ is $S^\infty$.
\end{lemma}

\begin{proof}
Define $A$ to be the interval of beliefs $p \in [0,1]$ such that there exist a prior $p_0\in [0,1]$, a pure reasonable SE $\sigma$, and a history $h$ satisfying $p^1(h)<p$, $p^2(h)<p$,
and $\sigma(h)=R$. 
Note that
$\inf A\geq \widetilde p >0$. It follows that $\phi(p)< \inf A$ whenever $p\in A$ is sufficiently close to $\inf A$.

We need to prove that $\inf A\geq p^*$. We argue by contradiction, and assume that $\inf A<p^*$. 
This implies the existence of $p<p^*$ such that $p\in A$ and $\phi(p)<\inf  A$.
We fix such a belief $p$, and let  $p_0$ and  $\sigma$ be given as in the definition of the set $A$. By assumption, the 
set $H_1=\{h \colon p^1(h)<p, p^2(h)<p, \sigma(h)=R\}$ is not empty.

The main part of the proof consists in showing that $h\in H_1$ implies $hS\in H_1$. 
The argument follows closely the proof of Lemma~\ref{prop A1}.

Let $h\in H_1$ be arbitrary. 
Since $h\in H_1$, the continuation play induced by $\sigma$ after $h$ starts with $R$ and ends with $S^\infty$. It can thus be written as $\overline h R S^\infty$, 
where $\overline h$ is either empty or starts with $R$.

We denote by  $i$ the player active at  $h\overline h$ and note that $\sigma(h\overline h)=R$.

\smallskip\noindent
\textbf{Part 1.} We prove that  $\overline h=\emptyset$. 
Assume to the contrary that $\overline h \neq \emptyset$, so that $\overline h$ starts with $R$. Then, the beliefs of both players $j$ following $h\overline h S$ are such that 
 $p^{j}(h\overline hS)\leq \phi(p^{j}(h))<\phi(p)$.\footnote{The first inequality holds since $n_{\mathrm e}(\overline h)\geq 1$, and the second since $h\in H_1$.}
Since $\phi(p)<\inf A$, one has $\sigma(h\overline h S)= S$.
Moreover, by the definition of $\overline h$, $\sigma(h \overline h R)=S$.
Hence,  the continuation play 
induced by $\sigma$ after $h\overline h a$ is $S^\infty$ for each $a\in \{R,S\}$.
Since $p^i(h\overline h)<p^*$, player $i$'s continuation payoff at $h\overline h$ is higher when choosing $S$. This contradicts $\sigma(h\overline h) =R$, 
and hence $\overline h=\emptyset$.

\smallskip\noindent
\textbf{Part 2.}
 We  prove that $\sigma(hS)= R$, which implies that $hS\in H_1$.
Assume to the contrary that $\sigma(hS)= S$. Since $\overline h =\emptyset$, the continuation  play 
induced by $\sigma$ after $h R$ is $S^\infty$. 
On the other hand, $\sigma(hS)=S$ implies that $i$ is facing at $h$ a version of the second thought experiment.\footnote{The assumption $\sigma(hS)=S$ implies that the outcomes of $j$'s most recent experiments are immediately disclosed if $i$ chooses $S$.} Since $p^1(h)<p^*$, player~$i$ is better off deviating to $S$ at $h$, contradicting $\sigma(h)= R$. 

\medskip 

We have thus proved that $hS\in H_1$ whenever $h\in H_1$. Let any $h\in H_1$ be given. Thus, $hS\in H_1$ as well, which implies in turn that 
$hSS\in H_1$. At the history $hSS$, 
 the most recent choice of each agent was $S$, 
and hence the continuation game is `isomorphic' to the initial game with prior $p'_0:=p^1(hSS)= p^2(hSS)$.
That is, the profile induced by $\sigma$ following $hSS$ is a sequential equilibrium of the entire game, in which the prior belief is
$p'_0$.  It then follows from Lemma~\ref{prop A1} that the play path induced by $\sigma$ after $hSS$ is $S^\infty$. In particular, 
$\sigma(hSS)=S$, which contradicts $hSS\in H_1$.
\end{proof}

\begin{lemma}\label{lemm last2}
Let $\sigma$ be a pure reasonable $\mathrm{SE}$. If $p_0<p^*_1$, then the path induced by $\sigma$ after $RS$ is $S^\infty$.
\end{lemma}

\begin{proof}
We will use the inequality $p_0< \min(p^*_2,\phi^{-1}(p^*))$, which follows from $p_0<p^*_1$.

We note that $p^1((RS)^n)= \phi^n(p_0)$ for each $n$, hence $N:=\min\{n\geq 1: \sigma((RS)^n)= S\}$ is well defined, finite, 
and at least $1$.
Since $p^i((RS)^N\cdot S)\leq \phi(p_0)<p^*$
for $i=1,2$,
by Lemma \ref{prop A2} the continuation play induced by $\sigma$ after $(RS)^N\cdot S$ is $S^\infty$. Hence, we only need to prove that $N=1$. This is implied by the  following two 
contradictory
claims.

\smallskip\noindent
\textbf{Claim:} If $N\geq 2$, one has $\sigma((RS)^{N-1}R)= R$. 

Assume that $N\geq 2$ and $\sigma((RS)^{N-1}R)= S$, so that the equilibrium continuation play at $(RS)^{N-1}$ is $RS^\infty$. Since $p^1((RS)^{N-1})\leq \phi(p_0)<p^*$, the induced payoff for player 1 is negative --- a contradiction.

\smallskip\noindent
\textbf{Claim:} If $N\geq 2$, one has $\sigma((RS)^{N-1}R)= S$. 

Assume that $N\geq 2$ and that $\sigma((RS)^{N-1}R)= R$. By the definition of $N$, one has $p^2((RS)^{N-1}R)=p_0$. This implies that $p^i((RS)^{N-1}\cdot RR)<p^*$ for both players, hence the continuation play induced by $\sigma $ following  $(RS)^{N-1}\cdot RR$ is $S^\infty$, by Lemma \ref{prop A2}. On the other hand, the continuation play in the event where player 2 deviates to $S$ at  $(RS)^{N-1}\cdot R$ is also $S^\infty$. Since the belief of player 2 at $(RS)^{N-1}\cdot R$ is below $p^*_2$, deviating to $S$ is profitable. 
\end{proof}
\section{Proof of Theorem \ref{th private1}}\label{sec proof private1}

\subsection{Proof of the first statement}

We argue by contradiction, so let $(\sigma,\pi)$ be a reasonable SE such that $\prob_\sigma$-a.s., the limit belief of each player $i$ satisfies $p^i_\infty>\widehat p$.  Since the sequence of beliefs of player $i$ forms a bounded martingale (w.r.t. to the filtration of his private histories), this implies that $p^i(h)>\widehat p$ for every on-path sequence of choices $h$.

Let  $\widetilde h$ be such that $\prob_\sigma(\widetilde h\mid \theta= B)>0$. 
By Theorem~\ref{prop private}, 
players eventually switch to $S$ forever if $\theta = B$, hence there is an on-path extension of $\widetilde h$ of the form $hSS$. Since 
$\sigma$ is a reasonable SE, one has $ \phi^{n_{\mathrm e}(h)}(p_0)= p^i(hSS)> \widehat p \geq \widetilde p\geq \phi^{\widetilde N}(p_0)$
and since 
  $n_{\mathrm e}(\widetilde h)\leq n_{\mathrm e}(h)$,
it follows that $n_{\mathrm{e}}(\widetilde h)\leq \widetilde N$.

Hence, $N_{\mathrm{e}}$ is bounded.
This implies the existence of some history $h$, with active player $i$, such that  $\prob_\sigma(hR\mid \theta = B)>0$ and $\sigma$ induces $S^\infty$ after $hR$.
At the history $h\!\cdot\! RS$,  player $i$ faces the same trade-offs as in the first thought experiment. Since by assumption $p^i(h\!\cdot\! RS)>\widehat p$, playing $R$ and then $S$ is a profitable deviation of player $i$ from $\sigma$.

\subsection{Proof of the second statement}

We argue by contradiction, so let a reasonable SE $(\sigma,\pi)$ be such that $p^i(h)<\widehat p$ and $\sigma(R\mid h)>0$ for some  $h$ with active player $i$. We denote by $\underline{p}$ the infimum of $p^i(h)$ over such histories and players. Below, we fix such a history $h$ (and active player $i$) such that $\phi(p^i(h))<\underline{p}$. 
This implies that the continuation play induced by $\sigma$ after $hR$ is of the form $R^kS^\infty$ for some $k\geq 0$.

We recall that  $q^i(h)$ is the probability that $j$ was successful, conditional on $\theta= G$, if the past sequence of choices is $h$ and $i$ was not successful.
Therefore, $p^i(h) q^i(h)$ is the probability that $i$ assigns at $h$ to the event that $j$ was successful.

\begin{claim}
    The continuation payoff of player $i$ at $h$ when choosing $S$ is at least
    \begin{equation}\label{cps}
\delta p^i(h) g \times q^i(h),
\end{equation}
\end{claim}

\begin{proof}
If $j$'s last choice along $h$ is $S$, then $q^i(h)=0$ and the claim holds. Assume thus that $j$'s last choice along $h$ is $R$. 
This implies that $p^j(hS)\leq \phi(p^i(h))< \underline{p}$, and hence $\sigma(hS)= S$. Thus, player $j$ chooses $R$ at $hS$ if and only if he was successful along $h$, which has probability $p^i(h)q^i(h)$. 
Hence, the continuation play induced by $\sigma$ after $h\!\cdot\! SR$ is $R^\infty$, and $i$'s expected continuation payoff is $g$. Since $i$'s expected continuation payoff after $h\!\cdot\! SS$ is nonnegative, this concludes the proof of the claim.
\end{proof}
\smallskip

We will prove that the expected continuation payoff of $i$ when choosing $R$ at $h$ is strictly lower than (\ref{cps}). For convenience, we assume w.l.o.g. that $i$ assigns probability 1 to $R$ at $h$.\footnote{Formally, let $\widetilde \sigma_i$ be the strategy of $i$ that coincides with $\sigma_i$, except that $\widetilde \sigma_i(h)=R$. As $\sigma_i(R\mid h)>0$, and as $\sigma_i$ is sequentially rational at  $h$ against $\sigma_j$, the expected continuation payoff of player $i$ when facing $\sigma_j$ is the same, whether he uses $\sigma_i$ or $\widetilde \sigma_i$.}

As in the main text, we denote $\omega_i:=\inf\{n\geq 1: \sigma(S\mid h\!\cdot\! (RR)^{n})>0\}$
and $\omega_j:=\inf\{n\geq 0: \sigma(S\mid h\!\cdot\! (RR)^n\!\cdot\! R)=1\}$.
\begin{claim} One has $\omega_i<+\infty $ or $\omega_j<+\infty$ (or both).
\label{claim:19}
\end{claim}

That is, following $hR$, player $i$ will eventually choose the safe arm with positive probability in the event that player $j$ keeps experimenting, or player $j$ will eventually pick the safe arm for sure. 
\medskip

\begin{proof}[Proof of Claim~\ref{claim:19}]
Assume that $\omega_i= +\infty$. Then,  the experiments of player $i$ following $h\!\cdot\! RR$ are uninformative as long as player $j$ keeps choosing $R$, so that $p^j(h\!\cdot\! (RR)^n\!\cdot\! R)= \phi^n(p^j(hR))$ for each $n \geq 0$. Hence  $p^j(h\!\cdot\! (RR)^n\!\cdot\! R)< \widetilde p$ for $n$ large enough, which implies that $\sigma(h\!\cdot\! (RR)^n\!\cdot\! R)=S$ for $n$ large and therefore, $\omega_j<+\infty$.
\end{proof}

\medskip
We write $\omega:=\inf(\omega_i,\omega_j)$. 
\medskip

\smallskip\noindent
\textbf{Case 1:} $\omega= \omega_j< \omega_i$.

In this case, when starting from $h$, player $i$ chooses $R$ until the first random time where player $j$ chooses $S$, which occurs after at most $\omega$ periods.  

For $k\in  \llbracket 0, \omega\rrbracket$, denote by $\pi_k$ the probability that $j$ first chooses $S$ in the $k$-th period that follows $h$ conditional on $\theta= B$.\footnote{Thus, $\pi_0= \sigma(S\mid h\!\cdot\! R)$, $\pi_1= \sigma(S\mid h\!\cdot\! (RR)\!\cdot\! R)\times \sigma(R\mid h\!\cdot\! R)$, etc.}

For $k\in  \llbracket 0, \omega\rrbracket$, let $\sigma_j(k)$ be the strategy of $j$ that (i)   chooses the risky arm $k$ times when starting from $h$, then switches to $S$, and (ii) coincides with $\sigma_j$ elsewhere. We note that the distribution of continuation plays after $h$ induced by the behavior strategy profile $\sigma$ is a mixture of the distributions induced by $(\sigma_i,\sigma_j(k))$, with weights $\pi_k$, $k\in \llbracket 0,\omega\rrbracket$.

We denote by $\gamma^i(k)$ the continuation payoff of player $i$ 
at $h$
when facing $ \sigma_j(k)$, so that the expected continuation payoff of player $i$ at $h$ is equal to $\displaystyle \sum_{k=0}^\omega \pi_k \gamma^i(k)$.

\begin{claim}\label{claim gamma}
    For each $k\in \llbracket 0,\omega\rrbracket$, $\gamma^i(k) <\delta g  p^i(h) q^i(h)$.
\end{claim}

    Since $\delta g p^i(h) q^i(h)$ is a lower bound on $i$'s continuation payoff when choosing $S$ at $h$, this implies that $i$'s continuation payoff is strictly lower when playing $R=\sigma(h)$ than when playing $S$, a contradiction.
\medskip

The proof will make use of the following observation on a variant of the first thought experiment. 
If $p^i(h)< \widehat p$, choosing $R$ \emph{once} to encourage \emph{one} additional experiment is not worth it. Deceiving $j$ for an extended duration is \textit{a fortiori} not worth it either. Indeed, consider a variant of the thought experiment in which $i$ expects that choosing $R$ for $k$ times in a row will induce $j$ to experiment $k$ times as well. Assume that $k=2$ for simplicity. At the history $h\cdot RR$, the belief of $i$ is $\phi(p^i(h))$, because $j$'s most recent experiment is uninformative. In addition, the marginal benefit from experimenting once is lower than when at $h$, because choosing $S$ will give access to the outcome of $j$'s experiment. These two observations imply that $i$ would rather choose $S$ when at $h\!\cdot\! RR$. The algebra (for general $k$) is straightforward.

\medskip

\begin{proof}[Proof of Claim~\ref{claim gamma}]
The flow reward of $i$ is 
$p^i(h) \lambda m -c$ in each  period $t \in  \llbracket0, k\rrbracket$ after $h$, and 
the expected continuation payoff reward from period $k+1$ onwards is $g$ if some player was successful, and 0 otherwise. 

Thus, the overall continuation payoff  $\gamma^i(k)$  is given by
\begin{eqnarray}
&&\left(p^i(h) \lambda m -c\right) (1-\delta^{k+1})  \label{ctnK1} \\
&+&\delta ^{k+1} p^i(h) g \left\{\left(1-\left(1-\lambda\right)^{k+1}\right) +(1-\lambda)^{k+1} \left(q^i(h) +(1-q^i(h))\times \left(1-(1-\lambda)^{k}\right)\right)\right\}. \nonumber 
\end{eqnarray}
The first term between braces is the probability that $i$ is successful while the second term is the probability that $j$ was successful (but not $i$), either along $h$ or after $h$.

We claim that the expression in (\ref{ctnK1}) is strictly lower than (\ref{cps}). This follows from three observations.
\begin{enumerate}
\item
The difference between (\ref{cps}) and (\ref{ctnK1}) is increasing in $q^i(h)$, so we only need to prove that the inequality holds when $q^i(h)=0$, that is, 
\begin{eqnarray}
   && \left(p^i(h) \lambda m -c\right) (1-\delta^{k+1})  \label{ctnKbis}\\
   &+&\delta ^{k+1} p^i(h) g\left\{\left(1-\left(1-\lambda\right)^{k+1}\right) +(1-\lambda)^{k+1} \left(1-(1-\lambda)^{k}\right)\right\}< 0. \nonumber
\end{eqnarray}

\item
The choice of player $i$ in period $k$ does not affect $j$'s subsequent choice. Since $i$'s belief at $h\cdot(RR)^k$ is $\phi^k(p^i(h))<p^*$, the continuation payoff in (\ref{ctnKbis}) is lower than if $i$ chooses $S$ in that period. The latter payoff is given by
\begin{equation}\label{ctnK21}
    \left(p^i(h) \lambda m -c\right) (1-\delta^{k})+\delta ^{k} p^i(h) g\left\{\left(1-\left(1-\lambda\right)^{k}\right) +\delta (1-\lambda)^{k}  \left(1-(1-\lambda)^{k}\right)\right\}.
\end{equation}

\item
The quantity in (\ref{ctnK21}) is the payoff of player $i$ in the variant of the first thought experiment discussed before the proof.
Since
$p^i(h)< \widehat p$, this payoff is \emph{strictly} negative.
\end{enumerate}\end{proof}

\smallskip\noindent
\textbf{Case 2:} $\omega_j\geq \omega= \omega_i\geq 1$.

This is a variant of the first case. When starting from $h$, player $i$ chooses $R$ until the $\omega$-th period after $h$, or  until the first random time where player $j$ chooses $S$, whichever occurs first.
As in the first case, the continuation payoff of $i$ given $\sigma$ is a convex combination of the continuation payoffs $\gamma^i(k)$ induced by $(\sigma_i,\sigma_j(k))$, $k\in \llbracket 0,\omega\rrbracket$. 

For $k<\omega$, the expression of $\gamma^i(k)$ is still given by (\ref{ctnK1}), and $\gamma^i(k)<\delta g  p^i(h) q^i(h)$, as in Claim \ref{claim gamma}. Hence the result will follow if we prove that $\gamma^i(\gamma)< \delta g  p^i(h) q^i(h)$ as well. 

Under $(\sigma_i,\sigma_j(\omega))$,  players choose $R$ for $\omega$ periods, then $i$ switches to $S$. Hence $\gamma^i(\omega)$ is given by 
\begin{eqnarray}
&&\left(p^i(h) \lambda m -c\right) (1-\delta^\omega) \label{ctnN} \\
&+&\delta ^\omega p^i(h) g\left\{\left(1-\left(1-\lambda\right)^\omega\right) +\delta(1-\lambda)^\omega\left(q^i(h) +(1-q^i(h))\times \left(1-(1-\lambda)^{\omega}\right)\right)\right\}.\nonumber
\end{eqnarray}
 As above, the difference between (\ref{cps}) and (\ref{ctnN}) is increasing in $q^i(h)$, so we need only prove that 
 \begin{equation}\label{ctnNbis}
\left(p^i(h) \lambda m -c\right) (1-\delta^\omega)+\delta ^\omega p^i(h) g\left(1-\left(1-\lambda\right)^\omega\right) +\delta^{\omega+1} p^i(h) g(1-\lambda)^\omega\left(1-(1-\lambda)^{\omega}\right) <0.
\end{equation}
The result follows by noting that the left-hand side of (\ref{ctnNbis}) coincides with (\ref{ctnK21}).

\section{Proof of Theorem \ref{th private2}}\label{appendix:thm2022-1}

We recall that  $\sigma(h)=R$ if and only if $p^i(h)\geq \widehat p$, where $i$ is the player active at $h$.
We prove the sufficiency part.
We assume that $\phi^{n-1}(p_0)\geq p_n^*$, 
and  use the one-shot deviation principle. 

For $i=1,2$ and $h\in H$, denote by $d_i(h)$ the number of  $i$'s experiments after which he eventually chose the safe arm, thus \emph{disclosing} to $j$ that these experiments failed;
The difference
$u_i(h):=n_e^i(h)-d_i(h)$ is the number of \emph{un}-disclosed outcomes. That is, the last $u_i(h)$ choices of $i$ along $h$ were $R$. We will write $n_i, d_i, u_i$ unless there is a potential ambiguity.

It follows from the definition of $\sigma$ that $p^i(h)$ is either equal to $\phi^{n_i+d_j}(p_0)$ (if the last $u_j$ choices of $j$ are non-revealing), 
or to 1 (if the last $u_j$ choices of $j$ indicate a success of $j$).

We will use the following observation:
\begin{equation}\label{obs}
\widehat p\leq p^i(h)< 1 \Rightarrow n_i+d_j <n \Rightarrow p^j(h)<1,
\end{equation}
for each $h\in H$ where $i$ is the active player. The first implication holds since $p^i(h)$ is either 1 or $\phi^{n_i+d_j}(p_0)$. The second implication holds since as long as $n_i+d_j<n$, the strategy profile
$\sigma$ instructs $i$ to choose $R$, 
and hence the $u_i$ undisclosed experiments of $i$ (if any) are not informative to $j$.
\medskip

Fix $h\in H$, and let $i$ be the active player at $h$.
Given two infinite plays $h'$ and $h''$ in $\{S,R\}^\dN$, we write $h'\succeq_h h''$ if $i$ (weakly) prefers the continuation play $h'$ to $h''$. 
That is, consider the two strategy profiles $\sigma'$ and $\sigma''$ that coincide with $\sigma$ up to $h$, and then follow $h'$ and $h''$, respectively.
We say that $h'\succeq_h h''$ if $i$'s  expected continuation payoff,  computed using $i$'s belief at $h$, is higher under $\sigma'$ than under $\sigma''$.

\begin{lemma}
\label{lemma:3claims}
Let $h$ be given.
\begin{enumerate}
\item 
\label{claim1}
For each $k\geq 1$, 
\[ (RR)^kS^\infty \succeq_h S(RR)^k S^\infty\mbox{ if and only if }\phi^{k-1}\left(p^i(h)\right) \geq p^*.\]
\item
\label{claim2}
For each $k\geq 1$,  
\[(RR)^k S^\infty \succeq_h (RR)^{k-1} S^\infty  \mbox{ if and only if } \phi^{k-1}\left(p^i(h)\right) \geq \widehat p_{u_j+k-1}.\]\item
\label{claim3}
For each $k\geq 1$,  
\[(RR)^k RS S^\infty \succeq_h (RR)^k S^\infty  \mbox{ if and only if } \phi^{k}\left(p^i(h)\right) \geq p^*_{u_j+k}.\]\end{enumerate}
\end{lemma}

\begin{proof}
For each claim, the result is obvious if $p^i(h)=1$. We thus assume that $p^i(h)<1$. 

\smallskip
\noindent\textbf{Proof of 1.} The expected continuation payoff under $(RR)^kS^\infty$ is 
\[(1-\delta^k)\left( p^i(h)\lambda m -c\right) +\delta^k p^i(h) g\left\{\left(1-(1-\lambda)^k\right) 
+ \delta (1-\lambda)^k \left(1-(1-\lambda) ^{k+u_j}\right)\right\}.\]
The expected payoff under $S(RR)^{k}S^\infty$ is 
\[\delta\left\{(1-\delta^k)\left( p^i(h)\lambda m -c\right) +\delta^k p^i(h) g\left\{\left(1-(1-\lambda)^k\right)  + (1-\lambda)^k \left(1-(1-\lambda) ^{k+u_j}\right)\right\}\right\}.\]
Comparison of the two shows that 
$(RR)^kS^\infty \succeq_h S(RR)^k S^\infty$ if and only if 
\begin{equation}
\label{equ:905}
(1-\delta^k)\left( p^i(h)\lambda m -c\right) +\delta^k p(h) g\left\{\left(1-(1-\lambda)^k\right) \right\} \geq 0.
\end{equation}
The LHS of \eqref{equ:905} is the expected payoff of a single agent holding a prior of $p^i(h)$, who experiments for $k$ periods before switching to $S$. This proves the first claim.

\smallskip
\noindent\textbf{Proof of 2.} The continuation payoffs along the two plays $(RR)^kS^\infty$ and $(RR)^{k-1} S^\infty$
only differ in the event where player $i$ has no success along $(RR)^{k-1}$. 
At that point, player~$i$'s belief is $\phi^{k-1}(p^i(h))$, and player~$j$ has experimented $u_j+k-1$ times in sequence.
Hence player~$i$  prefers the continuation play $RR\cdot S^\infty$ to $S^ \infty$ if and only if $ \phi^{k-1}\left(p^i(h)\right) \geq \widehat p_{u_j+k-1}$, as claimed.

\smallskip
\noindent\textbf{Proof of 3.} 
The continuation payoffs along the two plays $(RR)^k\cdot RS\cdot S^\infty$ and $(RR)^k  S^\infty$
only differ in the event where player $i$ has no success along $(RR)^{k}$. 
At that point, player~$i$'s belief is $\phi^{k}(p^i(h))$, and  player $j$ has experimented $u_j+k$ in sequence.
Hence, player~$i$  prefers the continuation play $RS\cdot S^\infty$ to $S^ \infty$ if and only if $ \phi^{k}\left(p^i(h)\right) \geq  p^*_{u_j+k}$, as claimed.
\end{proof}

\bigskip
Let $h\in H$ be arbitrary, with active player $i$. We claim that $i$ has no one-step profitable deviation at $h$. This is clear if $p^i(h)=1$, so we assume that $p^i(h)<1$ and recall that $p^i(h)= \phi^{n_i+d_j}(p_0)$.

Assume first that $p^i(h)< \widehat p$, so that $\sigma(h)= S$. 
Either player $j$ was not successful in the past, in which case $p^j(hS)\leq p^i(h)<\widehat p$ and the continuation play is $S^\infty$; or player $j$ was successful and chooses $R$. Hence $p^i(h\!\cdot\! SR)=1$ and $i$'s continuation payoff is $\delta g$.

If instead $i$ deviates to $R$, the continuation play depends on $j$'s beliefs. 
If $p^j(hR)\geq \widehat p$, then $\sigma(hR)=R$ and hence $p^i(h\!\cdot\! RR)= \phi(p^i(h))<\widehat p$, 
which implies that $\sigma(h\!\cdot\! RR)= S$. Thus, deviating to $R$ triggers one additional experiment by $j$. Since $p^i(h)<\widehat p$, the deviation is not profitable. If instead $p^j(hR)<\widehat p$, the continuation play after $h\!\cdot\! R$ is the same as after $h\!\cdot\! S$ and hinges on whether $j$ was successful in the past. 
Since $p^i(h)<p^*$, deviating to $R$ is not profitable in that case either.
\medskip

We now assume that $p^i(h)\geq \widehat p$. The continuation play induced by $\sigma$ after $h$ is $R^tS^\infty$ for some $t\geq 1$, and  it is $SR^qS^\infty$ in case player~$i$ deviates to $S$,  for some $q\geq 0$. The exact values of $t$ and $q$ depend on $h$, as follows.

\begin{description}
\item[Equilibrium continuation:] The parity of $t$ depends on who stops experimenting. By (\ref{obs}),  $p^i(h),p^j(h)<1$. This implies that the first player to stop is player $j$ if $n_i+d_j<n_j+d_i$, and is player $i$ if $n_i+d_j\geq n_j+d_i$: 
\begin{itemize}
\item If $n_i+d_j<n_j+d_i$, the continuation play is $(RR)^k\!\cdot\! RS\!\cdot\! S^\infty$, with $k= n-(n_j+d_i)$.
\item If $n_i+d_j\geq n_j+d_i$, the continuation play is $(RR)^k\!\cdot\! S^\infty$, with $k=n-(n_i+d_j)$.
\end{itemize}
\item[Deviation continuation:] 
Player $j$ sees the choice of $S$ by $i$ as evidence that $i$ was not successful, hence $p^j(hS)=\phi^{n_1+n_2}(p_0)\leq p^i(h)$. 
Therefore, the first player to stop experimenting is player $j$ and he will experiment $\left(n-(n_1+n_2)\right)^+$ times: if player $i$ deviates to $S$ at $h$, the continuation play (deviation included) is 
\begin{itemize}
\item $S^\infty$, if $n_i+n_j\geq n$;
\item $SR\!\cdot\! (RR)^{k-1}\!\cdot\! RS\!\cdot\! S^\infty$, if $n_i+n_j<n$, with $n_i+n_j+k=n$.
\end{itemize}
\end{description}

We will prove in each case that the deviation to $S$ is not profitable, using the next lemma.

\begin{lemma}
One has: 
\begin{description}
\item[Q1.]  If $n_i+d_j< n_j+d_i$, then $(RR)^{n-(n_j+d_i)}\!\cdot\! RS\!\cdot\! S^\infty \succeq_h (RR)^{k}\!\cdot\! S^\infty$ for each $0\leq k\leq n-(n_j+d_i)$.
\item[Q2.] If $n_i+d_j\geq n_j+d_i$, then $(RR)^{n-(n_i+d_j)}\!\cdot\! S^\infty \succeq_h (RR)^{k-1}\!\cdot\! S^\infty$ for each $1\leq k\leq n-(n_i+d_j)$.
\item[Q3.] If $n_i+n_j<n$, then $(RR)^{n-(n_i+n_j)}\!\cdot\! S^\infty \succeq_h SR\!\cdot\! (RR)^{n-(n_i+n_j)-1}\!\cdot\! RS\!\cdot\! S^\infty $.
\end{description}
\end{lemma}

\begin{proof}
We start with Q1. Since the claim is empty if
 $n_j + d_i = n$, we assume $n_j + d_i < n$.

Since the sequences $(p^*_k)$ and $(\phi^k(p))$ are increasing and  decreasing, respectively, and since $\phi^{n-1}(p) \geq p^*_n$,
Lemma~\ref{lemma:3claims}(\ref{claim3}) implies that $(RR)^{n-(n_j+d_i)}\!\cdot\! RS\!\cdot\! S^\infty \succeq_h  (RR)^{n-(n_j+d_i)}\!\cdot\! S^\infty$. 
We argue below that, in addition,
$(RR)^k\!\cdot\! S^\infty \succeq_h (RR)^{k-1}\!\cdot\! S^\infty$  holds for each $1\leq k\leq n-(n_j+d_i)$.
This will imply \textbf{P1}.

By Lemma~\ref{lemma:3claims}(\ref{claim2}), it suffices to check that 
$\phi^{k-1}(p^i(h))\geq \widehat p_{u_j+k-1}$ for each $1\leq k\leq n-(n_j+d_i)$. 
The LHS of this inequality is decreasing in $k$, and the RHS is increasing in $k$. 
Hence, it suffices to check that it holds for $k= n-(n_j+d_i)$. 

For $k= n-(n_j+d_i)$, we have $\phi^{k-1}(p^i(h))= \phi^{n-(n_j+d_i)-1}\left(\phi^{n_i+d_j}(p_0)\right)\geq \phi^{n-2}(p_0)\geq \phi^{n-1}(p_0)$, and
 $\widehat p_{u_j+n-(n_j+d_i)-1}= \widehat p_{n-(d_i+d_j)-1}\leq \widehat p_n\leq p^*_n$, hence the desired inequality follows from the assumption $\phi^{n-1}(p_0)\geq p^*_n$.

\smallskip
\noindent\textbf{Proof of Q2.}
We proceed as in the second part of \textbf{Q1} and apply  Lemma~\ref{lemma:3claims}(\ref{claim2}) to show that $(RR)^k\!\cdot\! S^\infty \succeq_h (RR)^{k-1}\!\cdot\! S^\infty$ for each $1\leq k\leq n-(n_i+d_j)$. As before, the necessary and sufficient condition from Lemma~\ref{lemma:3claims}(\ref{claim2}) is most demanding when $k$ is highest.  So we only need to check that the condition is satisfied when $k= n-(n_i+d_j)$. 
For this $k$, the condition reduces to $\phi^{n-1}(p_0)\geq \widehat p_{u_j +n-(n_i+d_j)-1}$. Since $u_j+n-(n_i+d_j)-1\leq n$, and since $\widehat p_n\leq p^*_n$,
the inequality does indeed hold.

\smallskip
\noindent\textbf{Proof of Q3.}
Thanks to Lemma~\ref{lemma:3claims}(\ref{claim1}), it suffices to check that $\phi^{n-(n_i+n_j)-1}(p^i(h))\geq p^*$, 
which holds since $n_i+n_j\geq 0$ and $p^*\leq p^*_n$.
\end{proof}

\bigskip
We now prove that deviating to $S$ is not profitable.

\smallskip\noindent
{\textbf{Case 1}: $n_i+d_j<n_j+d_i$ and $n_i+n_j\geq n$}.

We need to prove that $(RR)^{n-(n_j+d_i)}\!\cdot\! RS\!\cdot\! S^\infty \succeq_h S^\infty$, which follows  from \textbf{P1} with $k=0$. 

\bigskip\noindent
{\textbf{Case 2: $n_i+d_j<n_j+d_i$ and $n_i+n_j< n$}}.

We need to prove that $(RR)^{n-(n_j+d_i)}\!\cdot\! RS\!\cdot\! S^\infty \succeq_h  SR\!\cdot\! (RR)^{n-(n_i+n_j)-1}\!\cdot\! RS\!\cdot\! S^\infty $, which  follows by applying \textbf{Q1} 
with $k=n-(n_i+n_j)$ and $\textbf{Q3}$.

\smallskip\noindent
{\textbf{Case 3}: $n_i+d_j\geq n_j+d_i$ and $n_i+n_j\geq n$}.

We need to prove that $(RR)^{n-(n_i+d_j)}\!\cdot\! S^\infty \succeq_h S^\infty$,
which follows from \textbf{Q2} with $k=1$. 

\smallskip\noindent
{\textbf{Case 4}: $n_i+d_j\geq n_j+d_i$ and $n_i+n_j< n$}.

We need to prove that $(RR)^{n-(n_i+d_j)}\cdot S^\infty \succeq_h S\cdot (RR)^{n-(n_i+n_j)}\cdot S^\infty$.
This follows by applying \textbf{Q2} 
with $k=n-(n_i+n_j)$ and \textbf{Q3}.

\section{Proof of Theorem \ref{th private3}}\label{sec proof th private3}

Theorem~\ref{th private3} will follow from Theorem \ref{th private2} and  simple algebra. 

\begin{lemma}\label{lemm nonvoid}
Set $\delta= \frac12$ and $\lambda = \frac{1}{n} $, with $n\geq 1$. Then $\phi(p^*_n)< \widehat p$.
\end{lemma}

The conclusion implies that for such $\delta,\lambda$, the interval $I:=(\max(\phi(\widehat p),\phi(p^*_n),\widehat p)$ is non-empty. For each $p_0\in \phi^{-n}(I)$, one has $n:=\inf\{k \geq 0: \phi^k(p_0)<\widehat p\}$ and 
$\phi^{n-1}(p_0)\geq p^*_n$, which implies that for such $p_0$, the strategy analyzed in Theorem \ref{th private2} is a reasonable SE with $N_{\mathrm e}=2n$ if $\theta=B$.

\medskip

\begin{proof}
By
Eqs.~\eqref{def:hatp} and~\eqref{def:p*n}, and since $\phi(p)= \frac{(1-\lambda)p}{1-\lambda p}$, the condition $\phi(p^*_n)<\widehat p$ reduces to
\[\frac{(1-\lambda) c(1-\delta)}{c(1-\delta)+g\left(1-\delta +\lambda \delta \left(1-\lambda\right)^n\right) -\lambda c(1-\delta)} <
\frac{c(1-\delta)}{c(1-\delta)+g\bigl(1-\delta + \delta\lambda(1+\delta-\delta \lambda)\bigr) },\]
which simplifies to 
\begin{equation}\label{condition}
(1-\lambda)\bigl(1-\delta + \delta\lambda(1+\delta-\delta \lambda)\bigr)
< (1-\delta)  +\lambda \delta (1-\lambda)^n.
\end{equation}

For $\delta=\frac12$, the inequality (\ref{condition}) 
reduces to
\[
(1-\lambda)\left(2+3\lambda -\lambda^2\right) < 2+ 2 (1-\lambda)^n\lambda,
\]
which simplifies to 
\begin{equation}\label{equ:delta12}
1-4\lambda +\lambda^2 < 2 (1-\lambda)^n.
\end{equation}
Since the LHS of (\ref{equ:delta12}) is at most $1-3\lambda$, and the RHS is at least $2(1-n\lambda)$, the inequality is satisfied whenever
$1-3\lambda< 2(1-n \lambda)$,
or $\lambda < \frac{2}{2n-3}$, and, in particular,
when $\lambda = \frac{1}{n}$.
\end{proof}

\begin{lemma}\label{lemm overexp}
Let $\eta>0$ be given. For all $n$ large enough, the cut-offs $\widehat p$ and $p^*_{\sqrt{\delta}}$ associated with $\delta= \frac12$ and $\lambda = \frac1n$ satisfy $\phi^{\lfloor\eta n\rfloor}(\widehat p)< p^*_{\sqrt{\delta}}$.
\end{lemma}

\begin{proof} Substituting $\delta=\frac12$ and $\lambda=\frac1n$ in Eq.~\eqref{def:p*}, we obtain
\[
p^*_{\sqrt{\delta}}  = 
\frac{c\left(\sqrt2-1\right)}{c\left(\sqrt2-1\right)+g\bigl(\sqrt2 - (1-\tfrac{1}{n})\bigr)},
\]
so that $\lim_{n\to +\infty} p^*_{\sqrt{\delta}}= \frac{c}{c+g}$.

On the other hand, 
$\phi^k(p)= \frac{(1-\lambda)^kp}{1-(1-(1-\lambda)^k) p}$ for each $p\in (0,1)$ and  $k\in \mathbb{N}$.
With $\delta = \frac12$ and $\lambda=\frac1n$,  given the expression of $\widehat p$ (see \eqref{def:hatp}), one has 
\[
\phi^k(\widehat p)
=\frac{c\left(2(1-\tfrac{1}{n})^k\right)}{g\left(2+\tfrac{3}{n} -\tfrac{1}{n^2}\right) +c \left(2(1-\tfrac{1}{n})^k\right)}.
\]
Substituting there $k=\lfloor \eta n\rfloor$, we get
\[\lim_{n\to +\infty} \phi^{\lfloor \eta n\rfloor}(\widehat p)= \frac{c e^{-\eta}}{g+c e^{-\eta}} < \frac{1}{1+g}.\]
It follows that $\phi^{\lfloor \eta n\rfloor}(\widehat p)< p^*_{\sqrt{\delta}}$ for every $n$ large enough, as desired. 
\end{proof}

\medskip
\begin{proof}[Proof of Theorem \ref{th private3}]
As noted after the statement of Lemma \ref{lemm nonvoid}, with $\delta= \frac12$, $\lambda= \frac1n$, and $p_0\in \displaystyle \left(\phi^{-n+1}(\widehat p),\phi^{-n}(\widehat p)\right]$, there is a pure SE with $N_{\mathrm e}= 2n$. 

Given $\eta>0$, for $p_0\in \displaystyle \left(\phi^{-n+1}(\widehat p),\phi^{-n}(\widehat p)\right]$
one has $\phi^{n+\lfloor\eta n\rfloor}(p_0)< \phi^{\lfloor \eta n\rfloor}(\widehat p)$, which is less than $p^*_{\sqrt{\delta}}$ for all $n$ large enough. 
This implies that $N^{**}\leq (1+\eta)n$.

For such an SE, one has $\displaystyle \frac{N_{\mathrm e}}{N^{**}}\geq \frac{2}{1+\eta}$. Since $\eta>0$ is arbitrary, this implies the result.
\end{proof}

\end{appendix}

\appendix
\setcounter{section}{5}


This appendix contains the proofs of some of the results stated in the main text.

\section{Proofs for public payoffs}

\subsection{Proof of Proposition \ref{prop public}}\label{sec proof public}

\paragraph{Remark.} As in the main text, a strategy profile is a map $\sigma:H\to\Delta(\{S,R\})$ where $H$ is the set of public histories, with the understanding that $\sigma(h)$ is the mixed choice when the sequence of past choices is $h$ and all experiments of the active player  at $h$ failed. Similarly, by  \emph{continuation path induced by a \emph{pure} profile}, we mean the sequence of  choices 
that is obtained when all experiments fail.

\subsubsection{Proof of \textbf{P1}}
    
    Fix a Nash equilibrium $\sigma$. We denote by $\prob_{\sigma}\left(\cdot \mid B\right)$ the distribution over plays induced by $\sigma$ conditional on $\{\theta= B\}$, and by $n_{\mathrm{e}}(h)$ the total number of experiments along $h$, for each history $h$. We emphasize that $n_{\mathrm{e}}$ is defined over the set of (finite) public histories while the total number of experiments $N_{\mathrm{e}}$ is defined over (infinite) plays.
    
    Since the active player chooses $S$ if the common belief is below $\widetilde p$, 
    on-path the total number of experiments $N_{\mathrm{e}}$ is  bounded by  
    $\widetilde N:=\min\{n \geq 0: \phi^n(p_0)< \widetilde p\}$. 
    This implies the existence of an on-path history $h$ such that the  player active at $h$ experiments with positive probability and $n_{\mathrm{e}}(h)+1= \max N_{\mathrm{e}}$. 
    Since that player is the  \emph{last} one to experiment,  his belief at $h$ is at least $p^*$, hence $n_{\mathrm{e}}(h)<N^*$, and, therefore,  $N_{\mathrm{e}}\leq N^*$, with probability 1.

    If, on the other hand, $\prob_{\sigma}(N_{\mathrm{e}}< N^* \mid B)>0$, 
    then  for each $\ep > 0$ there exists an  on-path history $h$, such that $n_{\mathrm{e}}(h)< N^*$ and such that the probability that someone will ever experiment after $h$ is at most $\ep$.
    At $h$, the expected continuation payoff of the active player is at most $\ep g$.
    On the other hand, since $n_{\mathrm{e}}(h)<N^*$, the continuation payoff from deviating to the one-player optimal strategy is bounded away from zero. 
    Hence, for $\ep > 0$ sufficiently small, there is a profitable deviation from $\sigma$.  This concludes the proof of \textbf{P1}.

\subsubsection{Proof of \textbf{P2}}

We now prove \textbf{P2} and assume for simplicity that $p_0<p_{\mathrm{myop}}$. We define a pure profile $\sigma$ inductively. Given $h$, we set $k(h):=N^*-n_{\mathrm{e}}(h)$, which we interpret as a remaining budget of experiments. Accordingly, we set $\sigma(h)=S$ whenever $k(h)\leq 0$.  Let now $h$ be such that $k(h)>0$ and denote by $i$ the player active at $h$, and by $j$ the other player. 
We set $\sigma(h)=R$ if  $j$'s previous choice is consistent with $\sigma$, or if the remaining budget $k(h)$ is even. Otherwise, we set $\sigma(h)=S$.


We argue that $i$ has no one-step profitable deviation. 
If $\sigma(h)=R$, then the  continuation sequence of choices under $\sigma$ is $R^{k(h)}:=\underbrace{RR\cdots R}_{k(h) \rm{~times}}$, at which point players switch to  $S$, unless some experiment was successful. 
If instead $i$ deviates to $S$, the continuation sequence of choices is  $SS\!\cdot\! R^{k(h)}$ if $k(h)$ is odd,
and  $SR^{k(h)}$ if $k(h)$ is even.
In each case, the deviation does not  modify the number of experiments of each player,  and only delays by one period the payoffs of $i$.

If $\sigma(h)= S$, then the continuation sequence of choices at $h$ is $SR\!\cdot\! R^{k(h)-1}$. If instead $i$ deviates to $R$, then the resulting sequence of choices is $RS\!\cdot\! R^{k(h)-1}$. While the total number of experiments is the same, the cost of experimentation has shifted towards $i$, making the deviation not profitable. 

\subsubsection{Proof of \textbf{P3}}
We first  prove the existence of a  continuum of SPE payoffs when $p_0\in (p^*, \min (p_{\mathrm{myop}},\phi^{-1}(p^*)))$.

By \textbf{P1}, there is exactly one experiment ($N^*=1$) in any SPE, and each player prefers that the cost of experimentation be borne by the other player since $p_0<p_{\mathrm{myop}}$. Specifically, if no player  has experimented so far, the active player prefers the continuation path $SR\!\cdot\! S^\infty$ to $RS\!\cdot\! S^\infty$.

As a result, there are exactly two pure SPEs: (i) the SPE $\sigma_1$ described in the proof of \textbf{P2} (whenever active,  player~1 plays $S$ if $R$ has been pulled at least once, and $R$ otherwise, while player~2 always plays $S$) and (ii)
an analogous SPE $\sigma_2$, obtained by exchanging  the roles of the players. 

There is an additional mixed SPE $\sigma^{*}$, in which
the active player experiments with probability $\alpha$
as long as $R$ was never chosen.
The value of $\alpha$ is such that the benefit of potentially having the other player carry the experiment cost compensates for the delay that is incurred, should the active player eventually be the one to experiment.

Consider finally the strategy profile $\sigma_y$ that coincides with $\sigma^*$, except for the fact that player 1 chooses $R$ with probability $y$ rather than $\alpha$ at the initial node.
The profile $\sigma_y$ inherits its SPE property from $\sigma^*$. Note that the equilibrium payoff of player~2 is increasing in $y$. The result follows.

\medskip

We turn to the uniqueness of symmetric Markov equilibria. As indicated in the text, we interpret the game as a stochastic game, whose state space is the set $P:=\{1\}\cup\{\phi^n(p_0),n\geq 0\}$ of the beliefs which are consistent with Bayesian updating and the prior $p_0$. The existence of a symmetric Markov equilibrium follows from the fact  that every symmetric stochastic game 
with countably many states and finite sets of actions has a symmetric equilibrium, see, e.g., \cite{parthasarathy1973discounted}.

Assume $f: [0,1]\to [0,1]$ is a symmetric Markov equilibrium, and denote by $\gamma_1(p)$ and $\gamma_2(p)$ the payoffs%
\footnote{Note that $\gamma_1(p)$ is typically different than $\gamma_2(p)$. 
Indeed, 
$\gamma_2(p)$ is player~$2$'s payoff when player~$1$ makes a choice,
while 
since the equilibrium is symmetric, $\gamma_1(p)$ is player~$2$'s payoff when she (player~$2$) makes a choice.}
of the two players when the prior is $p_0=p$. 
Note that 
\begin{equation}\label{R1}\gamma_1(p)\geq \delta \gamma_2(p),\mbox{ with equality if }f(p)<1.\end{equation} 
Indeed, since $f$ is symmetric,
the RHS $\delta\gamma_2(p)$ is the expected payoff of player 1 when deviating to $S$ in the first period and playing according to $f$ afterwards, so the inequality (\ref{R1}) follows from the equilibrium property.

We first claim that $f(p)>0$ if and only if $p>p^*$. Indeed, 
if $f(p_0)=0$ for some $p_0 > p^*$, then no player ever experiments when the prior is $p_0$, in contradiction with \textbf{P1}. Note next that if $f(p_0)>0$ for some $p_0 < p^*$, then $N_{\mathrm{e}}\geq 1$ when the prior is $p_0$, again in contradiction with \textbf{P1}.
It remains to show that $f(p^*)=0$. 
This follows from (\ref{R1}) since
(a) if $f(p^*)>0$, one has $\gamma_1(p^*)=0$,
and
(b) $\gamma_2(p^*)>0$.
Indeed, the equality $\gamma_1(p^*)=0$ holds since the overall payoff of player 1 when experimenting at $p^*$ is zero; $\gamma_2(p^*)>0$ since player 2 benefits from the fact player 1 experiments with positive probability.

We now prove the uniqueness claim by induction. 
Assume that for some $n\geq 0$, $f(p)$ is uniquely defined for every $p \in [0,\phi^{-n}(p^*)]$,\footnote{This holds for $n=0$.} and let 
 $p_0= p \in [\phi^{-n}(p^*),\phi^{-(n+1)}(p^*)]$.

The continuation payoff of player 1 when choosing $R$ in the first period is $\delta \gamma_2(\phi(p))$
if unsuccessful, 
and $g$ if successful. Since $f(p)>0$ and $\gamma_2(\phi(p))$ is uniquely defined,  the equilibrium payoff of player 1 is uniquely defined and is given by
\begin{equation}
\label{equ:900}
\gamma_1(p) 
= (1-\delta)\bigl( p\lambda m -c \bigr) +\delta \left( p\lambda  g+ (1-p\lambda) \gamma_2(\phi(p))\right).
\end{equation}

After player 1's first move, player 2 is in the position of the first player. Hence, denoting $\widetilde p$ the (random) belief of  player 2 after the first move of player 1, one has 
\begin{equation}
\label{equ:901}
\gamma_2(p) = \E\left[\gamma_1(\widetilde p)\right]=  f(p) \left( p \lambda g + (1-p\lambda) \gamma_1(\phi(p))\right) + (1-f(p))\gamma_1(p).
\end{equation} 

We claim that there is at most one equilibrium\footnote{Starting from $p_0=p$.} such that $f(p)<1$. Assume indeed that $f(p)<1$.  Then player 1 is indifferent at $p$ between $S$ and $R$, so that $\gamma_1(p)= \delta \gamma_2(p)<\gamma_2(p)$. Therefore, the right-hand side of (\ref{equ:901}) is (i) not constant in $f(p)$, since this would imply $\gamma_1(p)=\gamma_2(p)$, and (ii) increasing in $f(p)$, since $\gamma_2(p)> \gamma_1(p)$. 
This implies that there is \emph{at most} one value of $f(p)\in (0,1)$ such that both (\ref{equ:900}) and (\ref{equ:901}) hold, which proves our claim. 

In addition, we claim that if such a value $f(p)$ exists, then there cannot be another equilibrium such that $f(p)= 1$. Assume such an additional equilibrium exists, with equilibrium payoffs $\widetilde \gamma_1(p),\widetilde \gamma_2(p)$. 
Eq.~(\ref{equ:900}) still holds, hence $\widetilde \gamma_1(p)= \gamma_1(p)$. 
In addition, one must have $\widetilde \gamma_2(p)> \gamma_2(p)$, since (\ref{equ:901}) is increasing in $f$. 
Since $\gamma_1(p)= \delta \gamma_2(p)$, this implies $\widetilde \gamma_1(p) < \delta \widetilde \gamma_2(p)$: in the additional equilibrium, player 1 is better off deviating to the safe arm -- a contradiction. This concludes the proof of the uniqueness claim. 

\section{Proof for private outcomes}

\subsection{ Existence of a pure sequential equilibrium}\label{sec proof existence}

In this section, we formalize the sketch provided in Section \ref{sec pure SE}.

\subsubsection{Strategies}\label{sec strategies}

Fix $p_0\geq p^*$, and let $N^*:=\inf\{n \geq 0: \phi^n(p_0)< p^*\}$. 
Given $r \geq 0$, we set $h^*_r:=(RS)^r\cdot (RR)^{N^*-r}$, and we denote by $\gamma_r$ the expected payoff of player 2 in the scenario in which (i) both players follow $h^*_r$ in the first $N^*$ periods, (ii) in period $N^*+1$ player 1 chooses $S$ if he was not successful (and $R$ if successful), and (iii) player 2 then repeats forever the latter choice of player 1.\footnote{Unless player 2 was successful in the first $N^*$ periods.}
\smallskip

 The definition of $\sigma$ at a given history $h$ depends on the belief held by the player active at $h$. Since these beliefs will only depend on earlier play, there is no circularity in the definition.
\medskip

We set $\sigma(\emptyset)= R$, and first let $h$ be any public history that starts with $R$. If player 1 is active at $h$, we set $\sigma(h)= R$ if $p^1(h)\geq p^*$, and set $\sigma(h)= S$ otherwise. 
If player 2 is active at $h$, we set $\sigma(h)=S$ if either (i) $p^2(h)< p^*$ or if (ii) $h= (RS)^r\!\cdot\! R$ for some $r<N^*$ such that $\gamma_r < \max_{\tilde r\geq r} \gamma_{\tilde r}$. We set $\sigma(h)= R$ if neither (i) nor (ii) holds.
\medskip

Beliefs $p^i(h)$ are uniquely defined by Bayes rule at any $h$ such that $\prob_\sigma(h)>0$. Assuming players follow $\sigma$, the play unfolds as follows:
\begin{itemize}
\item The first $N^*$ periods form an experimentation phase. During this phase, the players follow the sequence $h^*_{r_0}$ of choices, where 
$r_0$ is the minimal integer that satisfies
$\gamma_{r_0}= \max_{r\in \llbracket 0,N^*\rrbracket}\gamma_r$. Player 1 experiments in every single period, and player 2 starts experimenting after some delay, until period $N^*$ included. 
The sequence of choices of a player conveys no information to the other,
and one's belief only incorporates one's failures (or successes).
\item In period $N^*+1$, players exchange information. At that point, player 1 holds the belief $p^1(h^*_{r_0}) = \phi^{N^*}(p_0)<p^*$ if unsuccessful, and chooses $S$. 
Thus, player 1's choice in period $N^*+1$ discloses the outcomes of his earlier experiments, and player 2 updates his belief to either $p^2(h^*_{r_0}\!\cdot\! R)=1$ or to $p^1(h^*_{r_0}\!\cdot\! S)= \phi^{2N^*-r_0}(p_0)<p^*$ (unless player 2 was successful). This implies that the continuation sequence of choices beyond period $N^*$, following $h^*_{r_0}$, is either $(RR)^\infty$ if player 1 was successful, $(SS)^\infty$ if no player was successful, or $(SR)\cdot(RR) ^\infty$ if player 2 was successful, but player 1 was not. The latter case is possible only if $r_0<N^*$.
\end{itemize}

\medskip

Finally, we define $\sigma$ at \emph{off-path} histories that start with $S$.
\begin{description}
    \item[Case 1] If $\gamma^1(\sigma)< \delta \gamma^2(\sigma)$, we set $\sigma(S)= S$, $\sigma(SR\!\cdot\! h')= \sigma(Rh')$, and $\sigma(SS\!\cdot\! h')= \sigma(h')$, for each $h'\in H$.
    \item[Case 2] If $\gamma^1(\sigma)\geq  \delta \gamma^2(\sigma)$, we set $\sigma(Sh')= \sigma(h')$ for each $h'\in H$.
\end{description}

\subsubsection{Beliefs}

  Since beliefs are not assumed to be reasonable, they are not uniquely pinned down by the definition of $\sigma$ at earlier nodes. 
    The definition of the beliefs $(p^1,p^2)$ relies on two core ideas:
\begin{itemize}
    \item Once a public history $h$ is reached where one player $i$ is convinced that the other was successful,   that is $p^i(h)=1$, then the belief of $i$ stays equal to 1 thereafter, while the beliefs of player $j$ are updated on the basis of $j$'s experiments, irrespective of $i$'s choices (which are therefore viewed as non-informative).
    \item When both beliefs $p^1(h)$ and $p^2(h)$ are below 1, they are updated using $\sigma$ as long as choices are  consistent with $\sigma$. 
Whenever 
the choice of
some player $i$ 
is inconsistent with
$\sigma$, the other player updates his belief to $p^j=1$  if $i$ had experimented at least once before deviating, and does not update his belief if $i$ had never experimented before deviating. 
\end{itemize}

Formally, we define beliefs inductively. 
Let $h\in H$ be arbitrary and assume that $p^1(h)$ and $p^2(h)$ have been defined. Let $i$ be the player active at $h$.

Player $i$'s belief after choosing $a$ at $h$ is set to (i) $p^i(ha)=p^i(h)$ if $a=S$  and (ii) $p^i(hR)=\phi(p^i(h))$.

Player $j$'s belief at $ha$ is defined as follows, as a function of $a$ and of the beliefs at $h$:
\begin{description}
    \item[Case 1:] $p^j(h)=1$. We set $p^j(ha)=1$ for each $a$.
    \item[Case 2:] $p^j(h)<1$ and $p^i(h)=1$. We set $p^j(ha)= p^j(h)$.
    \item[Case 3:] $p^i(h),p^j(h)<1$ and $a=\sigma(h)$. Then $p^j(ha)$ is uniquely pinned down from $p^j(h)$ by Bayes rule.
    \item[Case 4:]  $p^i(h),p^j(h)<1$ and $a\neq \sigma(h)$. 
We set $p^j(ha)=p^j(h)$ if $n^i_{\mathrm e}(h)=0$, and $p^j(ha)=1$ if $n^i_{\mathrm e}(h)\geq 1$.
\end{description}

\subsubsection{Sequential rationality}

We prove that for each history, the player active at that history has no profitable one-step deviation.

The construction of $\sigma$ in case player 1 deviates to $S$ in period 1 ensures that this deviation is not profitable. In addition, the sequential rationality of $\sigma$ at histories that start with $S$ will follow from the sequential rationality of $\sigma$ at histories that start with $R$.

Let then $h\in H$ be a non-empty history that starts with $R$, with active player $i$. 

If $p^i(h)=1$, the dominant choice is $R= \sigma(h)$. We assume below that $p^i(h)<1$. 

\smallskip

If $p^j(h)=1$, $i$ expects that future choices of $j$ will be non-informative, so that $i$ is effectively facing the one-player decision problem.\footnote{Sequential rationality dictates that $j$ will choose $R$. But we also need to check that $i$'s choice are rational after any sequence of choices.} Hence, it is optimal to choose $R$ if $p^i(h)\geq p^*$ and to choose $S$ if $p^i(h)\leq p^*$. If the active player is $i=1$, this shows that $i$'s choice at $h$ is sequentially rational. If the active player is $i=2$, the reasoning is slightly different. 
It follows from the construction of $\sigma$ that the only histories at which player 2 chooses $\sigma(h)=S$ while holding a belief $p^2(h)\in [p^*,1)$ are histories  of the form $(RS)^k\!\cdot\! R$ for some $k$. But these histories do not fall in that case since $p^1(h)<1$.

\smallskip

Assume finally that $p^i(h),p^j(h)<1$.

\paragraph{Case 1: $|h|< 2N^*$ and even.}

In that case, $h$ is a (strict) prefix of $h^*_r$ for some $r$. 
Choosing $R= \sigma(h)$ will 
induce $h^*_r$, at which point player 1  will learn the experiment outcomes of player 2.
Choosing $S$ induces player 2 to repeat $R$ forever since $p^2(hSh')=1$ for each $h'$. Since $p^1(h)\geq p^*$, the choice at $h$ of the risky arm is sequentially rational at $h$.

\paragraph{Case 2: $|h|< 2N^*$ and odd.}

In that case, $h$ is a (strict) prefix of $h^*_r$ for some $r$. That is, $h$ is either of the form $(RS)^k\!\cdot\! R$ for some $k< N^*$, or $h= (RS)^k\!\cdot\!
(RR)^n\cdot R$ for some $k,n$ with $n>0$.
If $h=(RS)^k\!\cdot\! R$, the overall payoff of player 2 is $\gamma_k$ if he chooses $R$, and is $\max_{\tilde r> k}\gamma_{\tilde r}$ if he chooses $S$. Hence the choice $\sigma(h)$ of player 2 is sequentially rational. 

If $h= (RS)^k\!\cdot\! (RR)^n\!\cdot\! R$, with $n>0$, choosing $S$ induces player 1 to repeat $R$ forever, since player 1 will update his belief to $p^1(h\!\cdot\! S)=1$. Since $p^2(h)\geq p^*$, the corresponding continuation payoff is at most the one-player optimum. On the other hand, choosing $R= \sigma(h)$ yields a continuation payoff that strictly exceeds the one-player optimum.

\paragraph{Case 3: $|h|= 2N^*$.}

Since $p^1(h),p^2(h)<1$, one has $h= h^*_r$ for some $r\in \llbracket 0,N^*\rrbracket$. 

The continuation sequence of choices under $\sigma$ is $S^\infty$, hence the expected continuation payoff of player 1 is nonnegative.\footnote{It need not be zero, because player 2 may have been successful in the experimentation phase.}
If player 1 chooses $R$, all future choices of player 2 are $R$, so that player 1 is facing the one-player problem. Since $p^1(h)<p^*$, his expected continuation payoff when choosing $R$ is negative.

\paragraph{Case 4: $|h| =2N^*+1$.}

Since $p^1(h),p^2(h)<1$, one has $h= h^*_r \!\cdot\! S$ for some $r\in \llbracket 0,N^*\rrbracket$. The continuation sequence of choices under $\sigma$ is $S^\infty$, hence the expected continuation payoff of player 2 is zero. If player 2 chooses $R$, all future choices of player 1 are noninformative: if $n^2_{\mathrm e}(h)\geq 1$ all future choices of player 1 will be $R$, while if $n^2_{\mathrm e}(h)=0$, future choices of player 1 will be $S$ until the next period in which player 2 chooses $R$ again (if any); Hence player 2 expects to be facing the one-player problem. Since $p^2(h)= \phi^{2N^*-r}(p_0)<p^*$, his expected continuation payoff when choosing $R$ is negative.

\paragraph{Case 5: $|h|> 2N^*$ is even.}

Note first that $p^1(h)\leq \phi^{N^*}(p_0)<p^*$, hence $\sigma(h)=S$. If instead player 1 chooses $R$, he expects all future choices of player 2 to be $R$, and therefore to be facing the one-player problem. Since $p^1(h)<p^*$, $\sigma(h)$ is sequentially rational at $h$.

\paragraph{Case 6: $|h|> 2N^*  + 2$ is odd.}

As in \textbf{Case 5}, we note 
that $p^2(h)\leq \phi^{N^*}(p_0)<p^*$, hence $\sigma(h)=S$. If instead player 2 chooses $R$, he expects all future choices of player 1 to be non-informative for the same reason as in \textbf{Case 4}. Therefore, player 2 expects to be facing the one-player problem. Since $p^2(h)<p^*$, $\sigma(h)$ is sequentially rational at $h$.

\subsubsection{Belief consistency}

To check the consistency of the beliefs with $\sigma$, we need to construct a sequence $(\tau_n)$ of completely mixed strategy profiles which converges to $\sigma$ (in the product topology) and such that the belief systems $(p_n)$ deduced from $\tau_n$ converge to $p$ (in the product topology). In particular, we need to allow for strategies that choose $S$ after a successful experiment. 

For clarity, we use the letter $\tau$ to denote strategies that sometimes play the safe arm after being successful, and keep the letter $\sigma$ for those strategies that repeat  the risky arm after a success. 

We proceed in two steps. 
For every $n \in \dN$
we define a 
strategy profile $\widetilde \tau_n= (\widetilde \tau^1_n,\widetilde \tau^2_n)$ as follows. Let $h$ be a public history with active player $i$, and $h^i$ be any private history of $i$ that is consistent with $h$. 
\footnote{In addition to the sequence of choices $h$, $h^i$ specifies  the outcomes of $i$'s experiments along $h$.}
\begin{itemize}
    \item If $p^i(h)<1$, we let $\widetilde \tau^i_n(h^i)$ assign probability $\frac{1}{n}$ to the arm $a\neq \sigma(h)$ if either $i$ was successful along $h^i$ or if $n^i_{\mathrm{e}}(h)=0$, and set $\widetilde \tau^i_n(h^i)= \sigma(h)$ otherwise.  Thus, player $i$ trembles if and only if he was successful in the past or never experimented.
    \item If $p^i(h)=1$, we let $\widetilde \tau^i_n(h^i)$ assign probability $1- \frac{1}{n}$ to the arm $\sigma(h)= R$ and probability $\frac{1}{n}$ to the other arm, irrespective of whether $i$ was successful along $h$.\footnote{That is, irrespective of the private history $h^i$ consistent with 
$h$.} Thus, player $i$'s trembles are non-informative.

\end{itemize}

In particular, at each public history, conditional on the information of the non-active player,  both arms are selected with positive probability by the active player. 
This implies that
all public histories occur with positive probability under $\widetilde \tau_n$, hence the beliefs $p^i_n(h)$ induced by $\widetilde \tau_n$ are uniquely defined 
by 
Bayes rule. 

One can verify that $\lim_{n\to +\infty}p^i_n(h)= p^i(h)$ for each $i$ and $h\in H$.

Given $n \in \dN$, 
we  let $(\tau_n^m)_{m \in \dN}$ be an arbitrary sequence of fully mixed strategies such that $\lim_{m \to \infty} \tau^m_n(h^i)= \widetilde \tau_n(h^i)$ for each $i$ and each private history $h^i$ of $i$. This implies that the beliefs induced by $(\tau^m_n)_{m \in \dN}$ converge to $p^i_n(h)$ as $m\to +\infty$. Using a diagonal extraction argument, this implies the existence of a sequence of fully mixed strategy profiles $(\tau_n)$,\footnote{With $\tau_n= \tau^{m_n}_n$ for some $m_n$.} such that the beliefs induced by $(\tau_n)$ converge to $p^i(h)$.

\subsection{Proof of Proposition~\ref{prop comparison p/e}}
\label{appendix prop 5}\label{section:prop:compare}

We follow the brief explanation in Section~\ref{section:simple:pure}. 
In the main text, we argued that $\sigma_0$ is a Nash equilibrium (NE) if $n=0$. Suppose then that $p_0 \geq p^*$ and $n \geq 1$.
Given  $n\geq 0$, we denote by  $\sigma_n$ the pure strategy profile that follows the sequence $h^*_\infty(n):=(RR)^{N^*+n}\!\cdot\!S^\infty$ as long as the past sequence of choices is consistent with $h^*_\infty(n)$, and that chooses $R$ otherwise. 
That is, $\sigma_n(h)=S$ if and only if  $h$ is a prefix of $h^*_\infty(n)$  of length $|h|\geq  2(N^*+n)$. 
In Lemma \ref{lemm cns}, we identify a necessary and sufficient condition on $n$ under which $\sigma_n$ is a Nash equilibrium. We then exploit this condition to prove Proposition \ref{prop comparison p/e}.

\medskip
\paragraph{Preparations.}
We first argue that there is  no profitable (unilateral) deviation that agrees on-path with $\sigma$ in the first $N^*+n$ periods. That is, given any on-path $h$ of length  $|h|\geq 2(N^*+n)$, the  player active at $h$ cannot profitably deviate in the continuation game. To show this, it suffices to consider three specific histories.

\paragraph{Case 1: $h= (RR)^{N^*+n}$.} 

If player 1 makes the equilibrium choice $S= \sigma(h)$, player 2's next choice truthfully reveals whether player 2 was successful in the first the $N^*+n$ periods. From that point on, there is no informational spillover from player 2. Indeed, following the history $h\!\cdot\!SS$, player 2 chooses $S$ as long as player 1 chooses the safe arm, and switches forever to $R$ in case player 1 experiments. Since $p^1(h\!\cdot\!SS)<p^*$, the best continuation strategy for player 1 in case he plays $S$ at $h$, is $\sigma$. On the other hand, if player 1 deviates to $R$ at $h$,
player 2 plays $R$ forever, irrespective of player 1's later choices. Since  $p^1(h)< p^*$, such a deviation is not profitable.

\paragraph{Case 2: $h= (RR)^{N^*+n}\!\cdot\!R$.} 

The history $h$ occurs if and only if player 1 is successful in the first $N^*+n$ periods. This implies that $p^2(h)=1$ and that player 2 expects player 1 to choose $R$ in all subsequent periods. The unique best-reply of player 2 is to follow $\sigma$ and choose the risky arm as well.

\paragraph{Case 3: $h= (RR)^{N^*+n}\!\cdot\!S$.} 

At $h$, player 2 assigns probability zero to the event that player 1 was successful. As in \textbf{Case 1},  player 2 expects player 1 to repeat $S$ as long as player 2 chooses $S$, and to  switch forever to $R$ in case player 2 experiments. Irrespective of player 2's continuation strategy, there is no informational spillover from player 1. Since $p^2(h)<p^*$, the claim holds in that case as well.\medskip

We can thus focus on on-path deviations that take place up to the period $N^*+n$.
If the active player deviates from $\sigma$ in one of these periods -- choosing $S$ rather than $R$ -- the other player sticks to the risky arm forever, hence the deviating player never infers any information on the outcomes of the non-deviating player's experiments.

The optimal 
deviation to $S$ among the first $N^*+n$ stages
consists in experimenting exactly $N^*$ periods until the deviator's belief falls below $p^*$, and then switching forever to the safe arm. 
Since $n\geq 1$,  this deviation coincides with $\sigma$ in the first $N^*$ periods and yields a continuation payoff of zero.
Hence,
$\sigma_n$ is a Nash equilibrium  if and only if the continuation payoff induced by $\sigma_n$ after $N^*$ (unsuccessful) experiments is non-negative for both players. 

We exploit this observation to prove the next result.

\begin{lemma}\label{lemm cns}
Denote $\overline p^*= \phi^{N^*}(p_0)\in (\phi(p^*),p^*)$ the belief of the players after $N^*$ periods.  
The strategy profile $\sigma_n$ is a Nash equilibrium if and only if
    \begin{equation}\label{eq cns}(1-\delta^n)\left(\overline p^*\lambda m -c\right)+\delta^n g \overline p^*\left\{\left(1-(1-\lambda)^n\right) +(1-\lambda)^n\left(1-\left(1-\lambda\right)^{n+N^*}\right)\delta \right\}\geq 0.\end{equation}
\end{lemma}

\begin{proof}
Assume that $n\geq 1$, and denote $\ell:=N^*+n$.
Consider player 1 first. 
Conditional on the public history $(RR)^{N^*}$, the expected flow payoff in each stage $N^*+1, \ldots, \ell$ is $\overline p^*\lambda m -c$. The continuation payoff from stage $\ell+1$ is equal to (i) $g$, if player 1 was successful, (ii) $\delta g$, if only player 2 was successful,
and (iii) 0, if both players were unsuccessful.

Hence, the expected continuation payoff of player 1 after $(RR)^{N^*}$ is 
\begin{equation}
\label{equ:cplayer 1}
CP_1:=(1-\delta^n)\left(\overline p^*\lambda m -c\right)+\delta^n g \overline p^*\left\{\left(1-(1-\lambda)^n\right) +(1-\lambda)^n\left(1-\left(1-\lambda\right)^{n+N^*}\right)\delta \right\}.
\end{equation}
%
Similarly,  the expected continuation payoff of player 2 after $h= (RR)^{N^*}\!\!\cdot\! R$ is
\[CP_2 : =(1-\delta^n)\left(\overline p^*\lambda m -c\right)+\delta^n g \overline p^*\left\{\left(1-(1-\lambda)^n\right) +(1-\lambda)^n\left(1-\left(1-\lambda\right)^{n+N^*}\right) \right\}.\]

There is a slight difference between $CP_1$ and $CP_2$ owing to the fact that player 2 learns the outcome of the $\ell$ experiments of player 1 \emph{before} playing in stage $\ell+1$, while player 1 learns whether player 2 was successful only \emph{after} playing in stage $\ell+1$.
Since $\delta < 1$, we have $CP_2\geq CP_1$. Consequently, $\sigma$ is a Nash equilibrium if and only if $CP_1\geq 0$. This is (\ref{eq cns}).
\end{proof}\\

\begin{corollary} \label{corJ}
Fix $\lambda$, $n\geq 0$, and $n^*\geq 0$ such that $(n+1)\lambda+(1-\lambda)^{2n+n^*+2}<1.$ If   $\delta$ is large enough, then $\sigma_n$ is a Nash equilibrium for all $p_0\geq \phi^{-n^*}(p^*)$.
 \end{corollary}

\begin{proof} 
Assume $p_0\geq \phi^{-n^*}(p^*)$. 
Then $N^*\geq n^*+1$.  Since $\overline p^*>\phi(p^*)$, 
Inequality  (\ref{eq cns}) is satisfied when
 \begin{equation} \label{eqJ2} \delta^n g \phi(p^*)\left\{\left(1-(1-\lambda)^n\right) +(1-\lambda)^n\left(1-\left(1-\lambda\right)^{n+n^*+1}\right)\delta \right\}\geq c(1-\delta^n).\end{equation}

Since $1-\delta^n \leq n(1-\delta)$,
Inequality (\ref{eqJ2}) holds as soon as 
\begin{equation} \label{eqJ21} \delta^n g \phi(p^*)\left\{\left(1-(1-\lambda)^n\right) +(1-\lambda)^n\left(1-\left(1-\lambda\right)^{n+n^*+1}\right)\delta \right\}> nc(1-\delta).\end{equation}
Since 
$\phi(p^*)=\frac{c(1-\delta)(1-\lambda)}{c(1-\delta)(1-\lambda)+g(1-\delta+\lambda\delta)}$,
taking the limit of \eqref{eqJ21} as $\delta$ goes to $1$,
we obtain that if 
\begin{equation}\label{N22} (1-\lambda)\left( 1 - (1-\lambda)^{2n+n^*+1}\right) > n\lambda, \end{equation}
then \eqref{eqJ2} holds for all $\delta$ sufficiently close to 1. Finally, note that the inequality (\ref{N22}) follows from the assumption.
\end{proof}







\begin{corollary}
Denote by $p_f:= \phi^{n-1}(\overline p^*)$ the players' belief \emph{prior} to their last experiment.
Then $p_f/p^*$ is arbitrarily close to   $1/e$, 
provided $\lambda$ is small enough and $\delta$ is close enough to 1.
\end{corollary}

\begin{proof}
Note that $\lim_{\lambda \to 0} (1-\lambda)^{\frac{1}{\lambda} - 4} = \frac{1}{e}$.
Let $\ep>0$ be arbitrary, and let $\lambda_0>0$ be such that 
$\displaystyle (1-\lambda)^{ \frac{1}{\lambda} -4 }< (1+\ep)\frac{1}{e}$, for each $\lambda<\lambda_0$. 
Fix $\lambda< \min(\lambda_0,\frac15)$, and set $n :=\lfloor \frac{1}{\lambda}\rfloor -2\geq 2$.  Since  $(n+1)\lambda <1$, by Corollary \ref{corJ}, for large $\delta$ there exists $\widetilde{p}_0<1$ such that $\sigma_n$ is a Nash equilibrium for  $p_0\geq \widetilde{p}_0$. 
Since $n-1\geq \frac1\lambda-4$, 
we have $(1-\lambda)^{n-1}<(1+\ep)\frac{1}{e}$. 
By definition,
$\bar p^*\leq p^*$.

Denote by $\displaystyle LR(p):= \frac{p}{1-p}$ the belief likelihood ratio as a function of the belief  $p\in (0,1)$ assigned to $\theta = G$. By Bayes rule, the beliefs $p^*$ and $p_f$ are related through 
\[\frac{LR(p_f)}{LR(p^*)}= (1-\lambda)^{n-1}\times \frac{LR(\overline p^*)}{LR(p^*)}<(1+\ep)\frac{1}{e}.\] 

For fixed $\lambda \in (0,1)$ and $n$, both $p^*$ and $p_f= \phi^{\r{n-1}}(\overline p^*)$ converge to zero as $\delta \to 1$. It follows that for $\delta$ high enough, one has
\[\frac{p_f}{p^*} < (1+\ep)\frac{1}{e} .\]
Since $N^*$ increases to $\infty$ as $p_0\to 1$,  we have  $\displaystyle \frac{N_{\mathrm{e}}}{N^*}=\frac{2(N^*+n)}{N^*}\xrightarrow[p_0\to 1]2$. 
Simple computations show that   $N^{**}-N^*$ is  bounded from above (by $1-\ln(2)/\ln(1-\lambda)$) when $\delta\to 1$, independently of $p_0$.  Hence $\displaystyle \frac{N_{\mathrm{e}}}{N^{**}}=\frac{N_{\mathrm{e}}}{N^{*}}\times\frac{N^*}{N^{**}}\xrightarrow[p_0\to 1]{}2.$\end{proof}
\color{black}

\bigskip

We turn to prove the claims of Proposition~\ref{prop comparison p/e}.
We first show that under the Nash equilibria $(\sigma_n)_{n \geq 0}$, the ratio $\frac{N_{\mathrm e}}{N^*}$ can be arbitrarily large.
Indeed,
fix $n \in \dN$.
For $\lambda>0$ small enough, $(n+1)\lambda+(1-\lambda)^{2n+2}<1$. 

In this case,
Corollary \ref{corJ} applies to   $n^*=0$.
In particular,  $\sigma_n$ is a Nash equilibrium for  $p_0=p^*$. At this equilibrium,
 $N^*=1$ and 
\[\frac{N_{\mathrm{e}}}{N^*}= \frac{2(N^*+n)}{N^*}= 2(1+n),\]
which can be  arbitrarily  large.

We next show that under the Nash equilibria $(\sigma_n)_{n \geq 0}$, $\frac{N_{\mathrm{e}}}{N^{**}}$ can be arbitrarily close to $\frac{2x_0}{\ln(2)}$, where $x_0\simeq 0.7968$ is the positive solution to $x+e^{-2x}=1$. 

Fix $\lambda>0$ small and $n \in\dN$ such that $(n+1)\lambda<x_0$. 
Applying Corollary \ref{corJ},
we obtain that $\sigma_n$ is a Nash equilibrium for  $p_0=p^*$. At this equilibrium,
\[ \frac{N_{\mathrm{e}}}{N^{**}}=\frac{2(1+n)}{N^{**}}, \]
where
$N^{**}=\inf\{l\geq 0, \phi^l(p^*)<p^{**}\}=\inf\left\{l \geq 0, (1-\lambda)^l< \frac{1-\delta+\lambda \delta}{1-\delta+\lambda \delta+\lambda \sqrt{\delta}}\right\}$.  
Hence,  for $\lambda$ small, $(1-\lambda)^{N^{**}}$ is close to $1/2$,
provided that $\delta$ is sufficiently close to $1$. 
We then have
\begin{equation}
\label{equ:J3}
\frac{N_{\mathrm{e}}}{N^{**}}=\frac{2(1+n)}{N^{**}}\sim \frac{2 (1+n)\lambda}{\ln(2)}.
\end{equation}
Since $(n+1)\lambda$ can be chosen arbitrarily close to $x_0$,  the ratio in \eqref{equ:J3} can be arbitrarily close to $\displaystyle\frac{2 x_0}{\ln(2)}\simeq 2.299$.
This concludes the proof of Proposition~\ref{prop comparison p/e}.\\

\subsection{Proof of Remark \ref{remark:6}}

We define a pure profile $\sigma$ and prove that it is a reasonable SE for a set of parameter values. The logic is simple: once some player has experimented at least once in the past, a player chooses $R$ if and only if he is convinced that the other was successful
(or, as usual, if he was successful). 
Formally:

\begin{itemize}
\item $\sigma(\emptyset)=R$: player 1 experiments in the first stage.
\item At any history of the form $h= Rh'$ with active player $i$, we set $\sigma(h)=R$ if $p^i(h)=1$, and $\sigma(h)=S$ otherwise. That is, the active player chooses $R$ if convinced that the other player was successful, and chooses $S$ otherwise. Since $p^i(h)$ is based on the definition of $\sigma$ at shorter histories, there is no circularity in this definition. 
\item The definition of $\sigma$ in the event where player 1 deviates to $S$ at the initial node  depends on the expected payoffs $\gamma^1(\sigma)$ and $\gamma^2(\sigma)$ induced by $\sigma$,
as in Section \ref{sec pure SE}: 
\begin{itemize}
    \item If  $\gamma^1(\sigma)>\delta \gamma^2(\sigma)$, players switch roles after $S$: $\sigma(Sh)=\sigma(h)$ for each $h$.
    \item If $\gamma^1(\sigma)\leq \delta \gamma^2(\sigma)$, $\sigma(S)=S$ and players resume according to $\sigma$ after $SS$: $\sigma(SS\!\cdot\!h)= \sigma(h)$ for each $h$.
\end{itemize} 
\end{itemize}

Finding a characterization of the histories $h$ such that $\sigma(h)=R$ without using beliefs is tricky, because the interpretation by $i$ of $j$'s previous choices hinges on how $j$ previously interpreted earlier choices of $i$.

\begin{proposition}\label{ex pure reasonable SE}  If $p_0\in [p^*,\min(p^*_1,\phi^{-1}(\widehat p)]$, then  $\sigma$ is a   reasonable SE.
\end{proposition}

\begin{proof}
We  check  that for each $h\in H$, the active player  at $h$ has no profitable one-step deviation. Given the symmetries in the construction, this follows from the three cases discussed below.

\paragraph{Case 1: $h=\emptyset$.}
If player 1 deviates to $S$ in the first stage, his expected payoff is equal to $\delta \gamma^1(\sigma)$ if $\gamma^1(\sigma)\leq \delta \gamma^2(\sigma)$,
and is equal to $\delta \gamma^2(\sigma)$ if $\gamma^1(\sigma)>  \delta \gamma^2(\sigma)$. In both cases, it is lower than $\gamma^1(\sigma)$.

\paragraph{Case 2: $h=S$.}
Assume first that  $\gamma^1(\sigma)\leq  \delta \gamma^2(\sigma)$. Then $\sigma(S)= S$ and player 2's payoff is equal to $\gamma^1(\sigma)$ if he deviates to $R$, and to $\delta \gamma^2(\sigma)$ if he chooses $S= \sigma(S)$.
Assume next that  $\gamma^1(\sigma)> \delta \gamma^2(\sigma)$. Then $\sigma(h)=R$ and player 2's payoff is equal to  $\delta \gamma ^2(\sigma)$ if he deviates to $S$, and to $\gamma^1(\sigma)$ if he chooses $R= \sigma(S)$. 
In both cases, player 2 is better off not deviating.

\paragraph{Case 3: $h= R\bar h$, for some $\bar h\in H$.}

 Assume first that $h=R$. The profile $\sigma$ induces the sequence $S^\infty$  from the second stage on,  irrespective of the choice of player 2 at $h$. That is, player 2 anticipates that he will learn in stage 2 the outcome of player 1's experiment. The assumption $p_0<p^*_1$ implies that the payoff when choosing $S=\sigma(h)$ is higher.
 
Assume next that $|h|>1$, and let $i$ be the player active at $h$. If $p^i(h)=1$, it is optimal to choose $R=\sigma(h)$.

We claim that otherwise, $p^i(R\bar h)\leq \phi(p_0)$. This is clear if the player active at $h$ is player 1, thanks to the inequality $p^1(h)\leq \phi^{n^1_{\mathrm{e}}(h)}(p_0)$.
If the player active at $h$ is player 2, then the assumption $|h|>1$ implies $|h|>2$. For such histories and by construction, the assumption $p^2(h)<1$ implies that player 1 has played $S$ at least once along $h$, implying that  $p^2(R\bar h)\leq \phi(p_0)$. According to $\sigma$, the continuation play following $h$ is $S^\infty$ if player 2 plays $S$, and $RR\!\cdot\!S^\infty$ if he plays $R$. Since $p^2(h)\leq \phi(p_0)<\widehat p$, the optimal choice is $S=\sigma(h)$.    
\end{proof}

\subsection{Non-pure reasonable SE}\label{sec: mixed}

When pure reasonable SEs fail to exist, players must resort to randomization. 
We here illustrate the nature of the equilibrium in such cases. 
We assume that $\phi(p^*)\in (\widehat{p}, \widehat{p}_1 )$, and that  $p_0 \in (p^*,\min\{p^*_1,   \phi^{-1}(\widehat{p}_1)\})$.  
These intervals are nonempty
when, e.g.,
$\delta - \lambda \geq 0.8$.
Since  $\phi(p^*_1)<p^*$ and $\phi(\widehat{p}_1)<\widehat{p}$, one has 
\[\phi^2(p_0)<\phi(\widehat{p}_1)<\widehat{p} <\phi(p^*)<\phi(p_0)<\phi(p^*_1)<p^*<p_0< p^*_1.\]

In this case, $N^*=1$ and simple algebraic manipulations show that $N^{**}=3$.

  \begin{proposition}\label{prop mixed}
      There is a unique reasonable sequential equilibrium outcome. At equilibrium, the support of the number $N_{\mathrm{e}}$ of experiments is $\{1,2,3\}$ if $\theta =B$. 
  \end{proposition}


\begin{proof}[Proof Sketch]
Define a  profile $\sigma$ as follows:
\begin{itemize}
    \item  $\sigma(h)= R$ if $h= (SS)^n$ for some $n\geq 0$ or if $h= (SS)^n\!\cdot\!S$ for some $n\geq 0$;
    \item  $\sigma(h)$ assigns probability $\alpha$ to $R$ whenever $h= S^nRS$ for some $n\geq 0$;
    \item  $\sigma(h)$ assigns probability $\beta$ to $R$ whenever $h= S^nRSR$ for some $n\geq 0$;
    \item for any other public history $h$ with active player $i$,  $\sigma(h)= S$ if $p^i(h)<1$,
    and $\sigma(h)= R$ if $p^i(h)$=1.
\end{itemize}

Figure 3 displays $\sigma$ and the active player's belief in the early periods of the game.  

\begin{center}
\begin{tikzpicture}
    
\node[draw,cyan] (T) at (0,0) {P1, $p_0$};
\node[draw,magenta] (2) at (2,4/3) {P2, $p_0$};
\node[draw,cyan] (3) at (4,0) {P1, $\phi(p_0)$};
\node[draw,magenta] (41) at (6,-4/3) {P2, $\phi(p_0)$};
\node[draw,magenta] (42) at (6,4/3) {P2, $p^*_2$};
\node[draw,cyan] (5) at (8,8/3) {P1, $\phi^2(p_0)$};
\node[draw,cyan] (6) at (8,0) {P1, $\phi^2(p_0)$};
\node (21) at (2,-4/3) {};
\node (31) at (4,-8/3) {};
\node (7) at (10,4/3) {$S^\infty$};
\node (8) at (10,-4/3) {$S^\infty$};
\node (51) at (8,-8/3) {$S^\infty$};

\draw (0.8,2/3) node[above]{$R$};
\draw (0.8,-2/3) node[below]{$S$};
\draw (3.2,2/3) node[above]{$S$};
\draw (4.8,2/3) node[above]{$R, \alpha$};
\draw (5.6,-2/3) node[above]{$S, 1-\alpha$};
\draw (7.6,2/3) node[above]{$S, 1-\beta$};
\draw (6.5,1.8) node[above]{$R, \beta$};

\draw[color=gray!40] (T) -- (21);

\draw[thick](3) -- (42);
\draw[thick](T) -- (2);
\draw[thick](2) -- (3);
\draw[thick](3) -- (41);
\draw[thick](41) -- (51);
\draw[thick](42) -- (5);
\draw[thick](42) -- (6);
\draw[thick](5) -- (7);
\draw[thick](6) -- (8);

\end{tikzpicture}

Figure 3: Strategies and beliefs in the example in Appendix~\ref{sec: mixed}. 
\end{center}

The value of $\alpha>0$ is pinned down by the condition that $p^2(RS\!\cdot\!R)= p^*_2$. The existence of $\alpha$ follows by continuity, since $p^2(RS\!\cdot\!R)= 1$ when $\alpha = 0$ and $p^2(RS\!\cdot\!R)= p_0$ when $\alpha= 1$.
The value of $\beta >0$ is pinned down by the condition that the continuation payoff of player 1 at $h= RS$ is zero, so that it is optimal to randomize at $h$. The existence of $\beta$ follows by continuity. Indeed, this continuation payoff is negative if $\beta = 0$ since $ \phi(p_0)<p^*$, and 
positive if $\beta = 1$ since $\phi(p_0)>\phi(p^*)>\widehat{p}$. 

\medskip

We now show that $\sigma$ is a reasonable SE when supplemented with the associated reasonable beliefs, and that all reasonable SEs induce the same distribution over plays.

 Given $\alpha$, the probability  that player 2 observes $RS\!\cdot\!R$ is $(\lambda +(1-\lambda)\alpha)$ if $\theta= G$ and $\alpha$ if $\theta= B$. By Bayes' rule, the condition on $\alpha$ translates to  $p^*_2=\frac{p_0(\lambda+(1-\lambda)\alpha)}{\alpha+p_0\lambda(1-\alpha)}$, so that
\[\alpha=\frac{p_0\lambda(1-p^*_2)}{p_0\lambda(1-p^*_2)+(p^*_2-p_0)}\in (0,1).\]

\medskip

Define  $l(p):= \displaystyle 1-\delta+\lambda \delta -\frac{c(1-\delta)(1-p)}{pg} $, so that 
\[p=\frac{c(1-\delta)}{c(1-\delta)+g(1-\delta+\lambda \delta-l(p))}.\]
The reader can verify that 
$l(p)$ is increasing in $p$, with $l(p^*)=0$, $l(p^*_1)=\lambda^2\delta$, and $l(p^*(2))=\lambda^2 \delta (2-\lambda).$ 
Therefore,
$\alpha=\alpha(p_0)$ can  be expressed as a function of $l(p_0)$ by:
$$\alpha=\frac{\lambda(1-\delta+\lambda \delta(1-\lambda)^2)}{-l(p_0)+\lambda(1-\delta+3 \lambda \delta -3\lambda^2 \delta + \lambda^3\delta)}.$$


Let $h$ be an arbitrary  public history, and let $i$ be the player active at $h$. We prove that player $i$ has no one-step profitable deviation at $h$.

\paragraph{Case 1: $h= (SS)^n$ for some $n\geq 0$.}

The continuation play induced by $\sigma$ after $h$ is $RS^\infty$ if player 1 chooses $R=\sigma(h)$, and $(SS)\!\cdot\!RS^\infty$ if player 1 deviates to $S$. Since $p_0>p^*$, it is better not to delay the experiment, hence the optimal choice for $i$ is $R=\sigma(h)$.

\paragraph{Case 1 bis: $h= (SS)^n\!\cdot\!S$ for some $n\geq 0$.}

The continuation play induced by $\sigma$ after $h$ is $SR\!\cdot\!S^\infty$ if player 2 chooses $S=\sigma(h)$, and $RS^\infty$ if player 1 deviates to $R$. In the former case, player 2's continuation payoff is $A:=\delta \gamma^2(\sigma)= p_0\delta^3 g(1+\alpha(1-\lambda))$. In the latter case, his continuation payoff is 
$B:= \gamma^1(\sigma)= (1-\delta)\left(p_0\lambda m -c\right) +p_0\delta \lambda g$. 
The reader can verify that $A\geq B$ if and only if 
\[l(p_0)\leq \lambda \delta^3(1+\alpha(1-\lambda)).\]
Since $p_0\in (p^*,p^*_1)$, one has $l(p_0)\leq \lambda^2 \delta$. 
On the other hand, one can verify that the assumptions in the proposition imply that $\lambda \leq \delta^2$, so that $l(p_0)\leq \lambda \delta^3$.

\paragraph{Case 2: $n_{\mathrm{e}}(h)=1$.}
We need to discuss according to the position of the experiment along $h$.  

Assume first that $h$ ends with $SS$. According to $\sigma$, the continuation play is  $S^\infty$ if  player $i$ chooses $S=\sigma(h)$,
and $RS\!\cdot\!S^\infty$ if he deviates to $R$. Since $p^i(h)= \phi(p_0)<p^*$, it is optimal to choose $S=\sigma(h)$.

Assume next that $h= S^kRS$ for some $k\geq 0$. The definition of $\beta$ ensures that player $i$ is indifferent between $S$ and $R$.

Assume finally that $h= S^kR$ for some $k\geq 0$. 
If player $i$ chooses $S=\sigma(h)$, 
in the next stage player $j$ randomizes  and chooses $S$ with probability $1-\alpha$.
The expected continuation payoff of player $i$ is, therefore,
\[A:=p_0\delta^2g(\lambda+(1-\lambda ) \alpha \lambda).\]

On the other hand, the the continuation play following $hR$ is $S^\infty$. 
Hence, should player $i$ deviate to $R$ at $h$, he expects to learn immediately the outcome of the unique experiment of $j$, and  $i$'s continuation is $g$ if and only if one of the two experiments is successful.
The expected continuation payoff of player $i$ is therefore given by
\[B:= (1-\delta) (p_0\lambda m -c) +\delta p_0 g\left(\lambda +\lambda(1-\lambda)\right)= p_0 g \left(l(p_0)+\lambda \delta (1-\lambda)\right).\]
We need to prove that $A\geq B$, which is equivalent to 
\begin{equation}\label{eq1}
 l(p_0)+\lambda \delta (1-\lambda)\leq \lambda \delta^2(1+\alpha(1-\lambda)).
\end{equation}
      
      Since $p_0\leq p^*_1$, one has $l(p_0)+\lambda \delta (1-\lambda)\leq l(p^*_1)+\lambda \delta (1-\lambda)=\lambda \delta$. 
      On the other hand, since $p_0\geq p^*$, one has $\lambda \delta^2(1+\alpha(1-\lambda)) \geq \lambda \delta^2(1+\alpha^*(1-\lambda))$, where $\displaystyle \alpha^*:=\lim_{p_0\to p^*}\alpha(p_0)=\frac{1-\delta + \delta \lambda (1-\lambda)^2}{1-\delta (1-\lambda)^3} $.
Therefore,   (\ref{eq1}) will follow from the following inequality:
      \begin{equation} \label{eq2} \delta (1+\alpha^*(1-\lambda))\geq 1\end{equation}
The inequality (\ref{eq2}) can be rewritten as $P(\delta)\geq 0$, 
with \[P(\delta):= -\delta^2\left(\left(1-\lambda\right)+\left(1-\lambda\right)^4\right) +\delta \left(1+ \left(1-\lambda\right) + \left(1-\lambda\right)^3\right)-1.\]
Observe that $P$ is concave in $\delta$, and  $P(1)>0$. 
In addition, the reader can verify that the  assumption $\phi(p^*)\geq \widehat p$ is equivalent to $\delta(1-\lambda) \geq a:=\frac{-1+\sqrt{5}}{2}$. Since $P$ is concave in $\delta$, it is therefore sufficient to prove that $\psi(\lambda):= P(\frac{a}{1-\lambda})>0$ for each $\lambda \leq 1-a$. 
Basic algebraic manipulations show  that $\psi$ is increasing, with $\psi(0)= 3a - 2a^2 -1>0$. This concludes the proof in that case as well.

\paragraph{Case 3: $n_{\mathrm{e}}(h)=2$.}

Again, we discuss several cases according to the location of the experiments along $h$.
\begin{itemize}
    \item Assume $h= S^k RR$ for some $k\geq 0$. The continuation play is $S^\infty$ if $i$ chooses $S$, and $RRS^\infty$ if $i$ chooses $R$ since $p^j(S^k\!\cdot\!RR\!\cdot\!R)=1$. Since $p^i(h)= \phi(p_0)< \widehat p_1$, the sequentially rational choice is $S=\sigma(h)$.
    \item Assume $h= S^kRSR$ for some $k\geq 0$. By construction of $\sigma$, $p^i(h)=p^*_2$ and $i$ is indifferent between $S$ and $R$.
    \item Assume $h= S^k RS^nR$ for some $k\geq 0$ and $n\geq 2$. The continuation play is $S^\infty$ if $i$ chooses $S$, and $RS^\infty$ if $i$ chooses $R$. Since $p^i(h)= \phi(p_0)< p^*_1$, the sequentially rational choice is $S=\sigma(h)$.
    \item Assume $h$ ends with $S$. The continuation play is $S^\infty$ if $i$ chooses $S$, and $RS^\infty$ if $i$ chooses $R$. Since $p^i(h)= \phi(p_0)< p^*$, the sequentially rational choice is $S=\sigma(h)$.
\end{itemize}

\paragraph{Case 4: $n_{\mathrm{e}}(h)\geq 3$.}

If $p^i(h)=1$, then $\sigma(h)=R$ is the optimal choice at $R$. These include histories of the form
\begin{itemize}
    \item $h= \overline h R^k$, where (i) $\overline h$ is either empty or ends with $S$, and (ii) $k$ is odd. 
    \item $h= \overline h RSR$, where $n_{\mathrm{e}}(\overline h)\geq 1$.
\end{itemize}
      
 Assume now that $p^i(h)<1$, so that $\sigma(h)=S$. 
 \begin{itemize}
 \item If $h= \overline h S$, then $p^i(h)= \phi^{n_{\mathrm{e}}(h)}(p_0)\leq \phi^3(p_0)<\widehat p$. 
 Choosing $R$, the strategy profile $\sigma$ induces the continuation path $RRS^\infty$. Since $p^i(h) < \widehat p$, the choice $S$ is sequentially rational.
\item Suppose that $h= \overline hS R^n$, where $n$ is even. The last $n/2$ experiments of $j$ are uninformative, for two  reasons. 
First,  at the history $\overline hS R$, player $i$ (who is active since $n$ is even) assigns probability 0 to the event that $j$ was successful along $\overline h$. Second, later experiments are uninformative as well since $\sigma(\overline hSR^k)= R$ for odd $k>1$. This implies that $p^i(h) \leq \phi^{n/2}(p_0)$. If $i$ chooses $S$ at $h$, he learns the outcomes of the last $n/2$ experiments of $j$. If $i$ instead chooses $R$, this induces an additional experiment of $j$, and then $i$ learns all experiments' outcomes, with a one-period delay. This is the thought experiment that defines $\widehat p_{n/2}$. 
Since  $\phi^{n/2}(p_0)<\widehat p_{n/2}$, the sequentially rational choice at $h $ is $S$.
 \end{itemize}

\end{proof}

\end{document}